\documentclass[onecolumn,showpacs,preprint,amsmath,amssymb,superscriptaddress]{revtex4-1}

\usepackage{dcolumn}
\usepackage{bm}
\usepackage{hyperref}
\newcommand{\ignore}[1]{}
\usepackage{graphicx}
\usepackage{subfigure}
\usepackage{amsmath,amssymb,amsfonts,amsthm,color,latexsym}

\usepackage[export]{adjustbox}

\usepackage{tabstackengine}

\def\mb{\mathbf}

\def\bm{\boldsymbol}
\newcommand{\Rmnum}[1]{\uppercase\expandafter{\romannumeral #1\relax}}
\newcommand{\rmnum}[1]{\lowercase\expandafter{\romannumeral #1\relax}}
\def\mb{\mathbf}
\newcommand*\diff{\mathop{}\!\mathrm{d}}

\mathchardef\mhyphen="2D
\usepackage{amsmath}

\newtheorem{theorem}{Theorem}[section]

\newtheorem{proposition}[theorem]{Proposition}

\usepackage[normalem]{ulem}

\begin{document}
\preprint{}

\title{Data-driven construction of stochastic reduced dynamics encoded with non-Markovian features}
\author{Zhiyuan She}
\affiliation{Department of Computational Mathematics, Science \& Engineering, Michigan State University, MI 48824, USA}%
\author{Pei Ge}
\affiliation{Department of Computational Mathematics, Science \& Engineering, Michigan State University, MI 48824, USA}
\author{Huan Lei}
\thanks{leihuan@msu.edu}
\affiliation{Department of Computational Mathematics, Science \& Engineering, Michigan State University, MI 48824, USA}
\affiliation{Department of Statistics \& Probability, Michigan State University, MI 48824, USA}%


\begin{abstract}
One important problem in constructing the reduced dynamics of molecular systems is the accurate modeling of the non-Markovian behavior arising from the dynamics of unresolved variables. The main complication emerges from the lack of scale separations, where the reduced dynamics generally exhibits pronounced memory and non-white noise terms. We propose a data-driven approach to learn the reduced model of multi-dimensional resolved variables that faithfully retains the non-Markovian dynamics. Different from the common approaches based on the direct construction of the memory function, the present approach seeks a set of non-Markovian features that encode the history of the resolved variables, and establishes a joint learning of the extended Markovian dynamics in terms of both the resolved variables and these features. The training is based on matching the evolution of the correlation functions of the extended variables that can be directly obtained from the ones of the resolved variables. The constructed model essentially approximates the multi-dimensional generalized Langevin equation and ensures numerical stability without empirical treatment. We demonstrate the effectiveness of the method by constructing the reduced models of molecular systems in terms of both one-dimensional and four-dimensional resolved variables. 
\end{abstract}


\maketitle
\section{Introduction}
\label{sec:introduction}
Predictive modeling of multi-scale dynamic systems is a long-standing problem in many 
fields such as biology, materials science, and fluid physics. One essential challenge arises from the high-dimensionality; numerical simulations of the full models often show limitations in the achievable spatio-temporal scales. Alternatively, reduced models in terms of a set of resolved variables are often used to probe the evolution on the scale of interest. However, the construction of reliable reduced models remains a highly non-trivial problem. In particular, for systems without a clear scale
separation, the reduced dynamics often exhibits non-Markovian memory effects, where the analytic form is generally unknown. To close the reduced dynamics, existing methods are primarily based on the following two approaches. The first approach seeks various numerical approximations of the memory term by projecting the full dynamics onto the resolved variables based on frameworks such as the Mori-Zwanzig formalism \cite{Mori1965, Zwanzig1973} or canonical models such as the generalized Langevin equation (GLE) \cite{Zwanzigbook}. Examples include the t-model approximation \cite{ChHaKu02}, the Galerkin discretization \cite{Darve_PNAS_2009}, regularized integral equation discretization \cite{Lange2006}, the hierarchical construction \cite{LiXian2014,Panos_rMZ_2015,zhu2018faber,Hudson2018,Price_Stinis_PNAS_2021}, and so on. 
Recent studies \cite{ma2018model, Vlachas_Byeon_PRSA_2018, Harlim_JCP_2020, WANG_Hesthaven_RNN_JCP_2020} based on the recurrent neural networks \cite{HochSchm97} provide a promising approach to learn the memory term of deterministic dynamics. Yet, for ergodic dynamics, how to impose the coherent noise term compensating for the unresolved variables remains open. The second approach parameterizes the memory term by certain ansatz, e.g., the fictitious particle \cite{Dav_Voth_JCP_2015}, continued fraction \cite{Wall_1948_book,Mori1965b}, rational function \cite{Corless_2003_book}, such that the memory and the noise terms can be embedded in an extended Markovian dynamics \cite{Mori1965b,ceriotti2009langevin, baczewski2013numerical,Dav_Voth_JCP_2015,Lei2016,Jung_Hanke_JCTC_2017,Lee2019,russo2019deep,ma2019coarse,Grogan_Lei_JCP_2020}. In addition, non-Markovian models are represented by discrete dynamics with exogenous inputs in form of NARMAX (nonlinear autoregression moving average
with exogenous input) \cite{chorin2015discrete,Lin_Lu_JCP_2021} and SINN (statistics information neural network) \cite{Zhu_Tang_arxiv_2022} and parameterized for each specific time step. 
Despite the overall success, most studies focus on the cases with a scalar memory function. Notably, the reduced model of a two-dimensional GLE is constructed in Ref. \cite{Lee2019}. To the best of our knowledge, the systematic construction of stochastic reduced dynamics of multi-dimensional resolved variables remains under-explored.  

Ideally, to obtain a reliable reduced model, the construction needs to accurately retain the non-Markovian features, enable certain modeling flexibility (e.g., the dimensionality of the resolved variables) and adaptivity (e.g., the order of approximation), and guarantee the numerical stability and robustness. In a recent study, we developed a Petrov-Galerkin approach \cite{Lei_Li_JCP_2021} to construct the non-Markovian reduced dynamics by projecting the full dynamics into a subspace spanned by a set of projection bases in form of the fractional derivatives of the resolved variables.  The obtained reduced model is parameterized as extended stochastic differential equations by introducing a set of test bases. Different from most existing approaches, the construction does not rely on the direct fitting of the memory function. Non-local statistical properties can be naturally matched by choosing the appropriate bases, and the model accuracy can be systematically improved by introducing more basis functions to expand the projection subspace.  
Despite these appealing properties, the construction relies on the heuristic choices of the projection and test bases. Given the target number of basis, how to choose the optimal basis functions for the best representation of the non-Markovian dynamics remains an open problem. Furthermore, the numerical stability needs to be treated empirically. These issues limit the applications in complex systems with multi-dimensional resolved variables. 

In this work, we aim to address the above issues by developing a new data-driven approach to construct the stochastic reduced dynamics of multi-dimensional resolved variables. The method is based on the joint learning of a set of non-Markovian features and the extended dynamic equation in terms of both the resolved variables and these features.  Unlike the empirically chosen projection bases adopted in the previous work \cite{Lei_Li_JCP_2021}, the non-Markovian features take an interpretable form that encodes the history of the resolved variables, and are learned along with the extended Markovian dynamic such that they are \emph{optimal} for the reduced model representation.  In this sense, they represent the optimal subspace that embodies the non-Markovian nature of the resolved variables.  
The learning process enables the adaptive choices of the number of features and is easy to implement by matching the evolution of the correlation functions of the extended variables. In particular, the explicit form of the encoder function enables us to obtain the correlation functions of these features directly from the ones of the resolved variables rather than the time-series samples.  
The constructed model automatically ensures numerical stability, strictly satisfies the second fluctuation-dissipation theorem \cite{Kubo66}, and retains the consistent invariant distribution \cite{espanol2004statistical,noid_multiscale_2008}.

We demonstrate the method by modeling the dynamics of a tagged particle immersed in solvents and a polymer molecule. With the same number of features (or equivalently, the projection bases), the present method yields more accurate reduced models than the previous methods \cite{Lei2016,Lei_Li_JCP_2021} due to the concurrent learning of the non-Markovian features. 
More importantly, reduced models with respect to multi-dimensional resolved variables can be conveniently constructed without the cumbersome efforts of matrix-valued kernel fitting and stabilization treatment. This is well-suited for model reduction in practical applications, where the constructed reduced models often need to retain the non-local correlations among the resolved variables. It provides a convenient approach to construct meso-scale models encoded with molecular-level fidelity and paves the way towards constructing reliable continuum-level transport model equations \cite{Lei_Wu_E_2020,Lei_E_DeePN2_2022}. 

Finally, it is worthwhile to mention that the present study focuses on the model reduction of ergodic dynamic systems where the full or part of the resolved variables are specified as known quantities that either retain a clear physical interpretation (e.g., the tagged particle position), or are experimentally accessible (e.g., the polymer end-to-end distance, the radius of gyration). 
Another relevant direction focuses on learning the slow or Markovian dynamics from the complex dynamic systems where the resolved variables are unknown \emph{a priori}; we refer to Refs. \cite{Rohrdanz_Zheng_JCP_2011,Per_Noe_JCP_2013, Li_Ma_MS_2014, Krivov_JCTC_2013,Lu_Vanden_JCP_2014,Bitt_Sch_JNS_2018} on learning resolved variables that retain the Markovianity, Refs. \cite{Coifman_Kevrekidis_MMS_2008,Chiavazzo_Covino_PNAS_2017,Crosskey_Maggioni_MMS_2017,Ye_Yang_Maggioni_arxiv_2021, Feng_Gao_arxiv_2022, Zieli_Hesthaven_2022} on learning the slow dynamics on a non-linear manifold, and Refs. \cite{Giannakis_ACHA_2019,Klus_Schu_JNonlinear_2018,Mauricio_Schut_JCP_2018,Klus_Schutte_Physical_D_2020} on model reduction of the transfer operator. 


\section{Methods}
\label{sec:methods}
\subsection{Problem Setup}
Let us consider the full model as a Hamiltonian system represented by a $6N\mhyphen$dimensional phase vector $\mb Z = [\mb Q; \mb P]$, where $\mb Q$ and $\mb P$ are the position and momentum vectors, respectively. The equation of motion follows
\begin{equation}
\dot{\mb Z} = \mb S \nabla H(\mb Z),
\label{eq:full_model}
\end{equation}
where $\mb S = \begin{pmatrix} 0 & \mb I\\ -\mb I &0
\end{pmatrix}$ is the symplectic matrix, and $H(\mb Z)$ is the Hamiltonian function and initial condition is given by $\mb Z(0) = \mb Z_0$.  
For high-dimensional systems with $N \gg 1$, the numerical simulation of Eq. \eqref{eq:full_model} can be computational expensive. It is often desirable to construct a reduced model with respect to a set of low-dimensional resolved variables $\mb z(t) := \phi\left(\mb Z(t)\right)$ where $\phi : \mathbb{R}^{6N} \to \mathbb{R}^{m}$ represents the mapping from the full to the coarse-grained state space with $m \ll N$. With the explicit form of $H(\mb Z)$ and $\phi(\mb Z)$, the evolution of $\mb z(t)$ can be mapped from the initial values via the Koopman operator \cite{Koopman315}, i.e.,  $\mb z(t) = {\rm e}^{tL} \mb z(0)$, where the Liouville operator $L \phi(\mb Z) = -(\left(\nabla H (\mb Z_0) \right)^T \mb{S}\nabla_{\mb Z_0}) \phi(\mb Z)$ depends on the full-dimensional phase vector $\mb Z$. Using the Duhamel–Dyson formula, the evolution of $\mb z(t)$ can be further formulated in terms of $\mb z$ based on the Mori-Zwanzig (MZ) projection formalism \cite{Mori1965, Zwanzig1973}. However, the numerical evaluation of the derived model relies on solving the full-dimensional orthogonal dynamics \cite{ChHaKu02}, which can be still computational expensive. 

In practice, the resolved variables are often defined by the position vector $\mb Q$. The MZ-formed reduced dynamics is often simplified into the GLEs, i.e.,
\begin{equation}
\begin{split}
\dot{\mb q} &= \mb M^{-1} \mb p \\
\dot{\mb p} &= -\nabla U(\mb q) - \int_0^t \bm\theta(t - \tau) \dot{\mb q}(\tau) \diff \tau +  \mb{\mathcal{R}}(t),
\end{split}
\label{eq:GLE}
\end{equation}
where $\mb q \in \mathbb{R}^{m} $ is the so-called collective variables, $\mb M$ is the mass matrix, $U(\mb q)$ is the free energy function, $\bm \theta(t): \mathbb{R}^{+} \to \mathbb{R}^{m\times m}$ is a matrix-valued function representing the memory kernel, and $\mathcal{R}(t)$ is a stationary colored noise related to $\bm\theta(t)$ through the second fluctuation-dissipation condition \cite{Kubo1966}, i.e., $\left\langle \mb{\mathcal{R}}(t) \mb{\mathcal{R}}(0)^T \right\rangle = k_BT \bm \theta(t)$. Numerical simulation of Eq. \eqref{eq:GLE} requires the explicit knowledge of both the free energy $U(\mb q)$ and the memory function $\bm\theta(t)$. Several methods based on importance sampling \cite{Kumar_Kollman_JCC_1992,Darve_JCP_2001,Laio_Parrinello_PNAS_2002} and temperature elevation \cite{Tuckerman_JCP_2002,TAMD_CPL_2006,TAMD_JCP_2006} have been developed to construct the multi-dimensional free energy function. In real applications, the main challenge often lies in the treatment of the memory kernel $\bm\theta(t)$. In particular, for multi-dimensional collective variables $\mb q$,  the efficient construction of numerically stable matrix-valued memory function remains under-explored. 

In this study, we develop an alternative approach to learn the reduced model. Rather than directly constructing the memory function $\bm\theta(t)$ in Eq. \eqref{eq:GLE}, we seek a set of non-Markovian features from the full model, denoted by $\left\{\bm\zeta_{i}\right\}_{i=1}^n$, and establish a joint learning of the reduced Markovian dynamics in terms of both the resolved variables and these features, i.e., 
\begin{equation}
\diff \tilde{\mb z} = \mb g\left(\tilde{\mb z}\right) \diff t + \bm\Sigma \diff \mb W_t, 
\label{eq:reduced_model_general}
\end{equation}
where $\tilde{\mb z} := \left[\mb q; \mb p; \bm\zeta_{1}; \cdots; \bm\zeta_{n}\right]$ represents the extended variables and $\mb W_t$ represents the standard Wiener process. 
In principle, any such extended system would generally lead to a non-Markovian dynamics for the resolved variables $\mb z = \left[\mb q; \mb p\right]$. However, the essential challenge is to determine $\left\{\bm\zeta_{i}\right\}_{i=1}^n$ so that the non-local statistical properties of $\mb z$ can be preserved with sufficient accuracy. Also, the form of $\mb g(\cdot)$ and $\bm\Sigma$ will need to be properly designed such that the reduced model retains the consistent thermal fluctuations and density distribution. In particular, the introduction of auxiliary variables can be loosely considered as approximating the full-dimensional Koopman operator in a sub-space. However, different from Ref. \cite{Lei_Li_JCP_2021}, the features $\left\{\bm\zeta_{i}\right\}_{i=1}^n$ are not the empirically-chosen projection bases; instead, they are learned along with model equation \eqref{eq:reduced_model_general} for the \emph{best approximation} of the non-local statistics. This essential difference enables the present method to generate more accurate reduced model and be easy to implement for multi-dimensional resolved variables without empirical treatment for numerical stability.  

\subsection{Non-Markovian features and the extended dynamics}
\label{sec:non_markovian_features}
To illustrate the essential idea, let us consider a solute molecule immersed solvent particles. To construct a reduced model \eqref{eq:reduced_model_general} for the solute molecule, a main question is how to construct the auxiliary variables $\bm \zeta := \left[\bm\zeta_{1}; \bm\zeta_{2}; \cdots; \bm\zeta_{n}\right]$. Intuitively, $\bm \zeta_i$ should depend on the full-dimensional vector $\mb Z$ such that their evolution encodes certain unresolved dynamics orthogonal to the subspace spanned by $\mb z(t)$. A straightforward approach is to represent $\bm \zeta(t)$ as a function of $\mb Z(t)$, i.e., $\bm \zeta = \mb h(\mb Z)$. However, the direct construction of the general formulation $\mb h(\mb Z)$ would become impractical since the learning  generally involves sampling the individual solvent particles near the solute molecule; the computational cost could become intractable due to the high-dimensionality of $\mb Z$. 

To circumvent the above challenges, the key ascribes to formulate
$\bm \zeta (t)$ such that it properly encodes the unresolved dynamics of $\mb Z(t)$, and
meanwhile, can be easily sampled.
One important observation is that the history of $\mb p(t)$ naturally encodes the unresolved dynamics orthogonal to
$\mb z(t)$ and the values can be conveniently obtained. To see this, we note that the dynamics follows $\dot{\mb p} = L\mb p$ where the Liouville operator $L \phi(\mb Z) = -(\left(\nabla H (\mb Z_0) \right)^T \mb{S}\nabla_{\mb Z_0}) \phi(\mb Z)$ depends on the full-dimensional vector $\mb Z$. Therefore, it is possible to construct $\bm \zeta(t)$ by some encoder functions
in terms of the time history of $\mb p(t)$, i.e., $\mb p(\tau)$ with $\tau \le t$, such that they retain certain components orthogonal to $\mb z(t)$. This is somewhat similar to
the study \cite{Lei_Wu_E_2020} on modeling the non-local constitutive dynamics of non-Newtonian fluids by learning a set of features from the polymer configuration space. The main difference is that the present features $\bm\zeta$ are non-Markovian in the temporal space.

Accordingly, we define a set of non-Markovian features $\left\{\bm\zeta_{i}\right\}_{i=1}^n$ by 
\begin{equation}
\begin{split}
\bm\zeta_{i}(t) &= \int_0^{+\infty} \bm\omega_i(\tau )\mb v(t-\tau) \diff \tau \\
&\approx \sum_{j=1}^{N_w} \mb w_i(j\delta t) \mb v(t- j\delta t) 
\end{split}
\label{eq:nonlocal_features}   
\end{equation}
where $\mb v := \mb M^{-1} \mb p$ represents the generalized velocity, $\bm\omega_i:\mathbb{R}^{+}\to \mathbb{R}^{m\times m}$ represents the encoder function represented by $N_w$ discrete weights $\left\{\mb w_i(j \delta t)\right\}_{j=1}^{N_w}$ whose values will be determined later.

$\bm\zeta_{i}(t)$ can be loosely viewed as a generalized convolution over the history of $\mb v(t)$ whose evolution encodes the orthogonal dynamics of $\mb z(t)$. Therefore, it is possible to learn a set of $\bm\zeta_i(t)$ such that the joint dynamics in terms of both $\mb z(t)$ and $\bm\zeta_i(t)$ can be well approximated by the extended Markovian model \eqref{eq:reduced_model_general}. Moreover, the linear form in terms of $\mb v(t)$ ensures that the invariant density function of $\bm\zeta_{i}(t)$ retains the Gaussian distribution consistent with $\mb v(t)$. We can further impose a constraint such that $\mb v(t)$ and $\bm\zeta_i(t)$ are uncorrelated. This leads to an additional quadratic term in the energy function of the extended system, i.e., $W(\mb q, \mb p, \bm\zeta) = U(\mb q) + \frac{1}{2} \mb p^T \mb M^{-1} \mb p + \frac{1}{2} \bm \zeta^T \hat{\bm\Lambda}^{-1} \bm \zeta$, where $\hat{\bm\Lambda}$ represents the covariance matrix of the $\bm \zeta_1, \cdots, \bm\zeta_n$.  The reduced dynamics can be written as
\begin{equation}
\diff
\begin{pmatrix}
{\mb q}\\
{\mb p}\\
{\bm \zeta}
\end{pmatrix}
= \mb G\nabla W(\mb q, \mb p,\bm \zeta) \diff t
+ \bm \Sigma \diff {\mb W}_t,
\label{eq:SDE_simple}
\end{equation}
where the matrix $\mb G \in \mathbb{R}^{(2+n)m \times (2+n)m}$ takes the form
\begin{align}
\mb G = \begin{pmatrix}
\begin{matrix}
0
\end{matrix}
&
\begin{matrix}
\mb I &0 &\cdots &0
\end{matrix}\\
\begin{matrix}
-\mb I \\ 0 \\\vdots \\ 0
\end{matrix}
&\mb{J}
\end{pmatrix}
\begin{pmatrix}
\mb I
&
\begin{matrix}
0 & 0 &\cdots & 0
\end{matrix} \\
\begin{matrix}
0 \\ 0\\ \vdots \\0
\end{matrix}
&
\begin{bmatrix}
\mb I & \begin{matrix} &  & \end{matrix} \\
\begin{matrix}
&\\
&
\end{matrix}
 &\hat{\bm\Lambda}
\end{bmatrix}
\end{pmatrix}
%
\label{eq:G_matrix}.
\end{align}
The matrix $\mb J \in \mathbb{R}^{nm \times nm}$ further determines the statistical properties of the resolved variables and will be learned along with the non-Markovian features $\left\{\bm \omega_i(t)\right\}_{i=1}^n$ from the training samples as discussed in next subsection. Given the reduced model in form of Eqs. \eqref{eq:SDE_simple} and \eqref{eq:G_matrix}, the noise covariance matrix can be determined by 
\begin{equation}
\bm\Sigma \bm\Sigma^T = -\beta^{-1} (\mb J \bm\Lambda +  \bm\Lambda^T \mb J^T), 
\label{eq:noise_covariance}
\end{equation}
where $\beta = 1/k_BT$ and $\bm\Lambda = {\rm diag}(\mb I, \hat{\bm\Lambda})$. Given this choice, we can show that model \eqref{eq:SDE_simple} strictly satisfies the second-fluctuation dissipation theorem. Specifically, the embedded kernel in Eq. \eqref{eq:SDE_simple} takes the form 
\begin{equation}
\Tilde{\bm\theta}(t) = -\left(\tilde{\mb J}_{11} \delta(t) + \tilde{\mb J}_{12} {\rm e}^{\tilde{\mb J}_{22} t} \tilde{\mb J}_{21}\right),
\end{equation}
where $\tilde{\mb J}_{11} = [\tilde{\mb J}]_{1:m, 1:m}$, $\tilde{\mb J}_{12} = [\tilde{\mb J}]_{1:m,m+1:}$ and $\tilde{\mb J}_{21} = [\tilde{\mb J}]_{m+1:,1:m}$ are the sub-blocks of the matrix $\tilde{\mb J} := \mb J \bm\Lambda$. The colored noise $\Tilde{\mb{\mathcal{R}}}(t)$ in terms of $\mb p(t)$ is related to $\Tilde{\bm\theta}(t)$ by
\begin{equation}
\left\langle \Tilde{\mb{\mathcal{R}}}(t) \Tilde{\mb{\mathcal{R}}}(t')^T \right\rangle = -\beta^{-1} \left(\tilde{\mb J}_{12} {\rm e}^{\tilde{\mb J}_{22}(t-t')} \tilde{\mb J}_{21} + (\tilde{\mb J}_{11} + \tilde{\mb J}_{11}^T) \delta(t - t')\right)
\label{eq:general_2nd_fdt}
\end{equation}
with $t' < t$. Moreover, we can show that the reduce model retains the consistent invariant density function with the full model, i.e.,  
\begin{equation}
\rho_{\rm eq} \propto \exp\left[-\beta W(\mb q, \mb p, \bm \zeta)\right].    
\end{equation}
We refer to Appendix \ref{app:fdt} and \ref{app:invariant_density} for details.

We conclude this subsection with two remarks: (\Rmnum{1}) In principle, the mass matrix $\mb M$ further depends on $\mb q$. Ref. \cite{Lee2019} reports that the varying mass matrix plays a secondary effect on the reduced dynamics of the molecular system therein. A constant mass matrix is adopted in the present study; reduced models with the varying mass matrix can be constructed by introducing an additional term in the conservative force and will be considered in the future study. (\Rmnum{2}) The non-Markovian features $\left\{\bm\zeta_{i}\right\}_{i=1}^n$ in form of Eq. \eqref{eq:nonlocal_features} can be generalized to retain the state-dependence, e.g., 
$\bm\zeta_{i}(t) = \int_0^{+\infty} \bm\omega_i(\tau, \mb q(\tau) )\mb v(t-\tau) \diff \tau $, which leads to a reduced model with state-dependent non-Markovian memory. In this study, we demonstrate the proposed learning framework by constructing the reduced model \eqref{eq:SDE_simple} that approximates the standard GLE \eqref{eq:GLE} with state-independent memory function $\bm \theta(t)$. 
As shown in the numerical examples, although $\bm\theta(t)$ is not explicitly constructed, it is well approximated by 
the memory kernel embedded in the reduced model \eqref{eq:SDE_simple} by matching the evolution of the correlation matrices for both the resolved and the extended variables.   
The learning of reduced models with the heterogeneous memory term will be systematically investigated in the future study.

\subsection{Joint learning of the reduced dynamics}
\label{sec:learn_reduced_dynamics}
Construction of the above reduced models relies on the joint learning of the non-Markovian features \eqref{eq:nonlocal_features} in form of the encoder functions $\left\{\bm\omega_i(t)\right\}_{i=1}^n$ and the reduced dynamics \eqref{eq:SDE_simple}\eqref{eq:G_matrix} determined by the free energy $U(\mb q)$ and the matrix $\mb J$. In this study, we represent the multi-dimensional free energy $U(\mb q)$ using a neural network and parameterize it based on the constraint molecular dynamics \cite{frenkel2001understanding}; we refer to Appendix for details. Furthermore, the covariance of the noise term specified by Eq. \eqref{eq:noise_covariance} implies that $\mb J$ and $\bm \Lambda$ (i.e., the encoder functions $\bm\omega_i(t)$) need to satisfy the following constraint 
\begin{equation}
\mb J \bm \Lambda + \bm \Lambda \mb J^T \preccurlyeq 0.
\label{eq:noise_constraint}
\end{equation}

Directly imposing the condition \eqref{eq:noise_constraint} becomes cumbersome for the joint learning of $\mb J$ and $\bm\omega_i(t)$. Fortunately, this issue can be avoided by imposing an orthogonal constraint among the non-Markovian features, i.e., 
\begin{equation}
\begin{split}
\left[\hat{\bm \Lambda}\right]_{ij} &:= \beta \left\langle \bm\zeta_i, \bm\zeta_j \right\rangle \\
&= \beta  \sum_{k,k'} \left\langle \mb w_i(t-k\delta t)\mb v(k\delta t), \mb w_j(t-k'\delta t)\mb v(k'\delta t) \right\rangle \\ 
&= \delta_{ij} \mb I, \quad 1\le i, j\le n, 
\end{split}    
\end{equation}
where the inner product $\displaystyle \left\langle \mb f(\mb Z) , \mb g(\mb Z) \right\rangle = \int \mb f(\mb Z) \mb g(\mb Z)^T \rho (\mb Z) \diff \mb Z$ is defined with respect to the equilibrium density function of the full model $\displaystyle \rho (\mb Z) = e^{-\beta H(\mb Z)} / \int e^{-\beta H(\mb Z)} \diff \mb Z $. In addition, we also impose the orthogonal constraints such that $\bm\zeta$ and $\mb p$ are uncorrelated. Therefore, condition \eqref{eq:noise_constraint} can be transformed into seeking $\mb J$ s.t. $\mb J + \mb J^T \preccurlyeq 0$, and we represent $\mb J$ by 
\begin{equation}
\mb J = - \mb L \mb L^T + \mb J^A,
\label{eq:J_SPD}
\end{equation}
where $\mb L \in \mathbb{R}^{(n+1)m\times(n+1)m}$ is the lower-triangle matrix with positive diagonal elements and $\mb L \mb L^T$ represents the Cholesky decomposition of a symmetric positive-definite matrix. $\mb J^A$ represents an anti-symmetric matrix. Unlike the studies \cite{Mori1965b,  ceriotti2009langevin} based on the direct kernel approximation, we note that $\mb J$ takes a more general form and is not restricted to be diagonal or tri-diagonal.  

With the proper form of $\mb J$, we can cast the reduced dynamics into the evolution of the correlation matrices by multiply $\mb v(0)$ to both sides of Eq. \eqref{eq:SDE_simple}, i.e.,
\begin{equation}
\frac{\diff}{\diff t}
\underbrace{
\begin{pmatrix}
{\left\langle \mb M\mb{v}, \mb v(0)\right\rangle} \\
{\left\langle \bm\zeta_1, \mb v(0)\right\rangle} \\
\vdots \\
\left\langle {\bm\zeta}_{n}, \mb v(0)\right\rangle
\end{pmatrix}
}_{\mb C_1(t)}
+
\underbrace{
\begin{pmatrix}
{\left\langle \nabla U(\mb q), \mb v(0)\right\rangle} \\
0 \\
\vdots \\
0
\end{pmatrix}
}_{\mb C_0(t)}
=
\mb J
\underbrace{
\begin{pmatrix}
{\left\langle \mb{v}, \mb v(0)\right\rangle} \\
{\left\langle \bm\zeta_1, \mb v(0)\right\rangle} \\
\vdots \\
\left\langle {\bm\zeta}_{n}, \mb v(0)\right\rangle
\end{pmatrix}
}_{\mb C_2(t)},
\label{eq:correlation_evolution}
\end{equation} 
where the correlation matrices $\left\langle \bm\zeta_i(t), \mb v(0)\right\rangle$ can be directly obtained from the correlation matrix of the resolved variables $\left\langle \mb v(t), \mb v(0)\right\rangle$ given the encoder weights, i.e.,
\begin{equation}
\left\langle \bm\zeta_i(t),  \mb v(0)\right\rangle = \sum_{j=1}^{N_w} \mb w_i(t_j)\left\langle  \mb v(t-t_j), \mb v(0)\right\rangle,
\nonumber
\end{equation}
where $t_j = j\delta t$ and encoder weights
$\mb w_i(t_j)$ will be learned jointly. Therefore, we are able to construct $\mb J$ from the pre-computed, noise-free correlation matrices instead of the on-the-fly computation from the time-series samples of $\mb q$ and $\mb p$. The reduced model can be trained by minimizing the following loss function 
\begin{equation}
\begin{split}
L_{C} &=  \sum_{j=1}^{N_t}\left\Vert {\mb C_1}'(t_j) + \mb C_0(t_j) -
\mb J {\mb C_2}(t_j)\right\Vert^2 ~
L_{\Lambda} = \left\Vert \bm \Lambda - \mb I\right \Vert^2, \\
L &= \lambda_{\Lambda} L_{C} +  \lambda_{\Lambda} L_{\Lambda}, 
\end{split}
\label{eq:loss_function}
\end{equation}
where $\lambda_C$ and $\lambda_{\Lambda}$ are the hyperparameters. $t_j$ refers to the discrete time points and $N_t$ represents the total number of sample points of the correlation matrices obtained from the full model.  The loss term $L_C$ ensures that the non-local statistical properties of the resolved variables can be accurately preserved while the loss term $L_{\Lambda}$ ensures the aforementioned orthogonality among the non-Markovian features. To simulate the constructed model, we always set $\hat{\bm \Lambda} = \mb I$ such that $\mb J$ in form of Eq. \eqref{eq:J_SPD} strictly satisfies the semi-positive definiteness condition.
We emphasize that the non-Markovian encoder weights $\left\{\mb w_i(t_j)\right\}_{j=1}^{N_w}$ do not explicitly appear in the loss function. However, they are involved in the training process along with $\mb J$ since the correlation functions $\mb C_1$ and $\mb C_2$ further depend on the definition of $\bm \zeta_i$, i.e., they are concurrently learned for the best approximation of the extended Markovian dynamics of $\left[\mb q; \mb p; \bm \zeta\right]$.  As shown in Sec. \ref{sec:numerical_result}, this joint learning of both the non-Markovian features and the dynamic equations enables us to probe the optimal representation of the reduced models that leads to more accurate numerical results than the ones constructed by the pre-selected bases, and can be easily implemented for models with multi-dimensional resolved variables. In this study, we choose $N_t = 5000$ for all the numerical examples and choose $N_w = 1800$ for the one-dimensional reduced model and $1200$ for the four-dimensional reduced model, respectively.  The training is conducted by the ADAM optimization algorithm \cite{Kingma_Ba_Adam_2015} where $1000$ points are randomly selected per each training step.

We conclude this subsection with the following remarks: (\Rmnum{1}) Instead of Eq. \eqref{eq:correlation_evolution}, the reduced dynamics can be also cast into the evolution of the correlation matrices by multiplying $\mb q(0)$ to both sides of Eq. \eqref{eq:SDE_simple}. For the present study, we found that both formulations yield accurate reduced models. (\Rmnum{2}) Rather than learning the full sets of non-Markovian features, we can fix one of them as $\mb M \dot{\mb v} + \nabla U(\mb q)$; this ensures that the time-derivatives of correlation functions of the reduced model can accurately match the values of the full model near $t = 0$.  In the numerical examples presented in following Sec. \ref{sec:numerical_result}, all the reduced models are constructed with this choice. For simple notation, we set it to be the last feature. For example, the fourth-order reduced model is constructed using $4$ non-Markovian features, where $\bm\zeta_4$ is set to be $\mb M \dot{\mb v} + \nabla U(\mb q)$ and does not involve in the training process. 
(\Rmnum{3}) In principle, for reduced models of Hamiltonian systems, the form of matrix $\mb J$ can be further restricted to 
\begin{equation}
\mb J = - \rm{diag}(0, \hat{\mb L} {\hat{\mb L}}^T) + \mb J^A,
\label{eq:J_SPD_2}
\end{equation}
where $\hat{\mb L} \in \mathbb{R}^{nm\times nm}$ is a lower-triangle matrix. Eq. \eqref{eq:J_SPD_2} ensures that the embedded kernel in Eq. \eqref{eq:SDE_simple} does not contain the Markovian memory term (i.e., $\left(\mb J_{11} +\mb J_{11}^T\right) \delta(t)$). $\tilde{\bm \theta}(t)$ recovers the form of standard GLE and the second fluctuation-dissipation relationship shown in Eq. \eqref{eq:general_2nd_fdt} recovers the standard form, i.e., 
$\left\langle \Tilde{\mb{\mathcal{R}}}(t) \Tilde{\mb{\mathcal{R}}}(t')^T \right\rangle = \beta^{-1}\Tilde{\bm \theta}(t-t')$. In this study, both forms yield accurate numerical results; the contribution of the Markovian memory term constructed by Eq. \eqref{eq:J_SPD} is less than $1\%$.

\section{Numerical results}
\label{sec:numerical_result}
\subsection{A tagged particle in aqueous environment}
We demonstrate our method by modeling a tagged particle immersed in solvent particles. The particle interaction is governed by the pairwise force
\begin{align*}
    \mathbf{F}_{ij} (\mb Q_{ij}) = \begin{cases}
    f_0(1-Q_{ij}/r_c)\mathbf{e}_{ij}, \textrm{ }&Q_{ij}\le r_c\\
    0,&Q_{ij}>r_c
    \end{cases}
\end{align*}
where $\mb Q_i$ and $\mb Q_j$ are the positions of $i\mhyphen$th and $j\mhyphen$th particles. $\mathbf{Q}_{ij} = \mathbf{Q}_i-\mathbf{Q}_j$, $Q_{ij}=\|\mathbf{Q}_i-\mathbf{Q}_j\|$, and $\mathbf{e}_{ij} = \frac{\mathbf{Q}_{ij}}{Q_{ij}}$, and $r_c$ is the cut-off distance. The full system consists of $4000$ particles in a $10\times10\times10$ cubic box with periodic boundary condition along each direction. We set $f_0 = 50$, $r_c = 1$, and impose Nos\'{e}-Hoover thermostat with $k_BT = 0.5$.

\begin{figure}
\centering
	\subfigure[]{
	\centering
		\hspace*{0em}\includegraphics[scale=0.45]{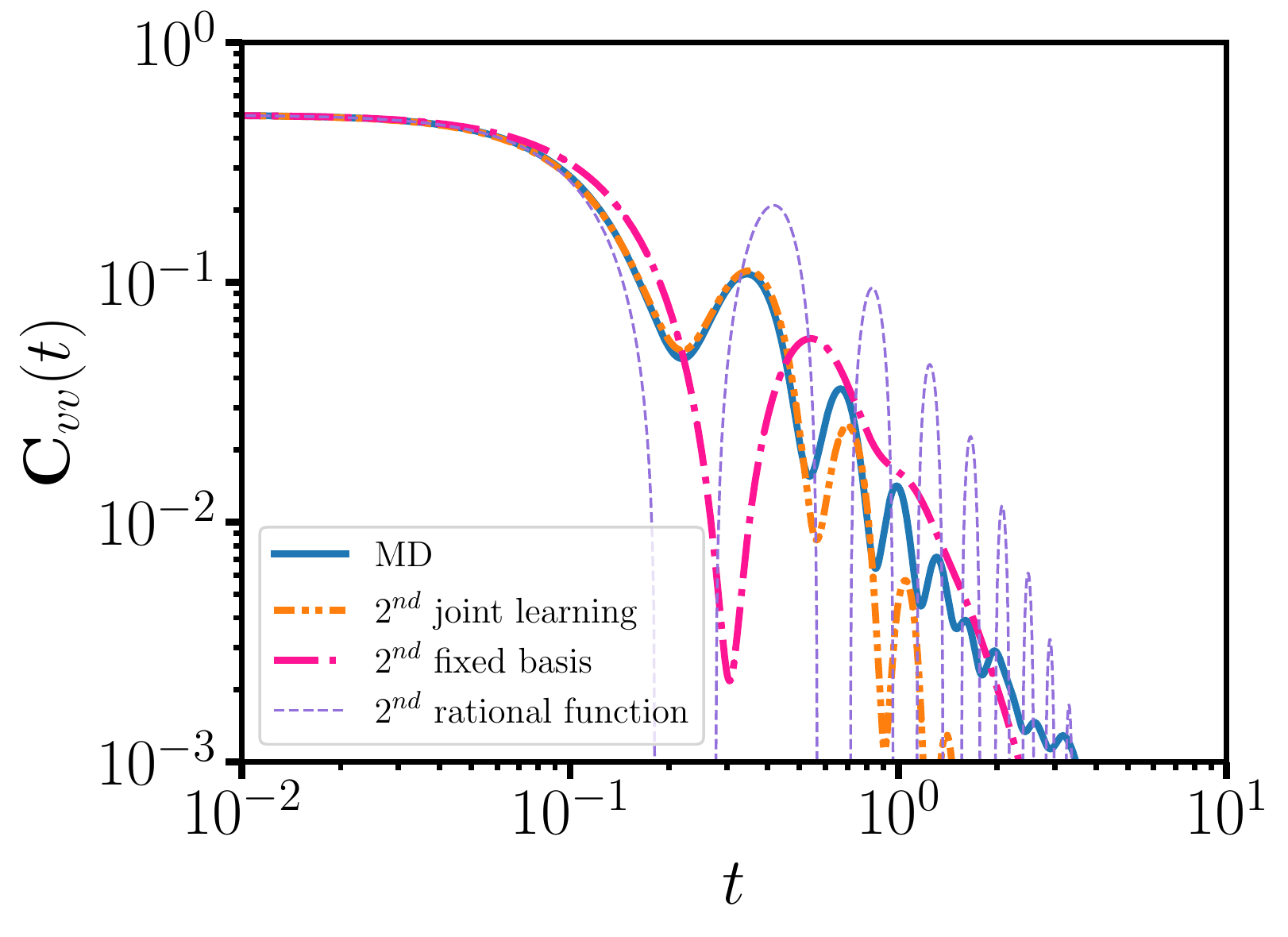}
	}
	\subfigure[]{
	\centering
		\hspace*{0em}\includegraphics[scale=0.45]{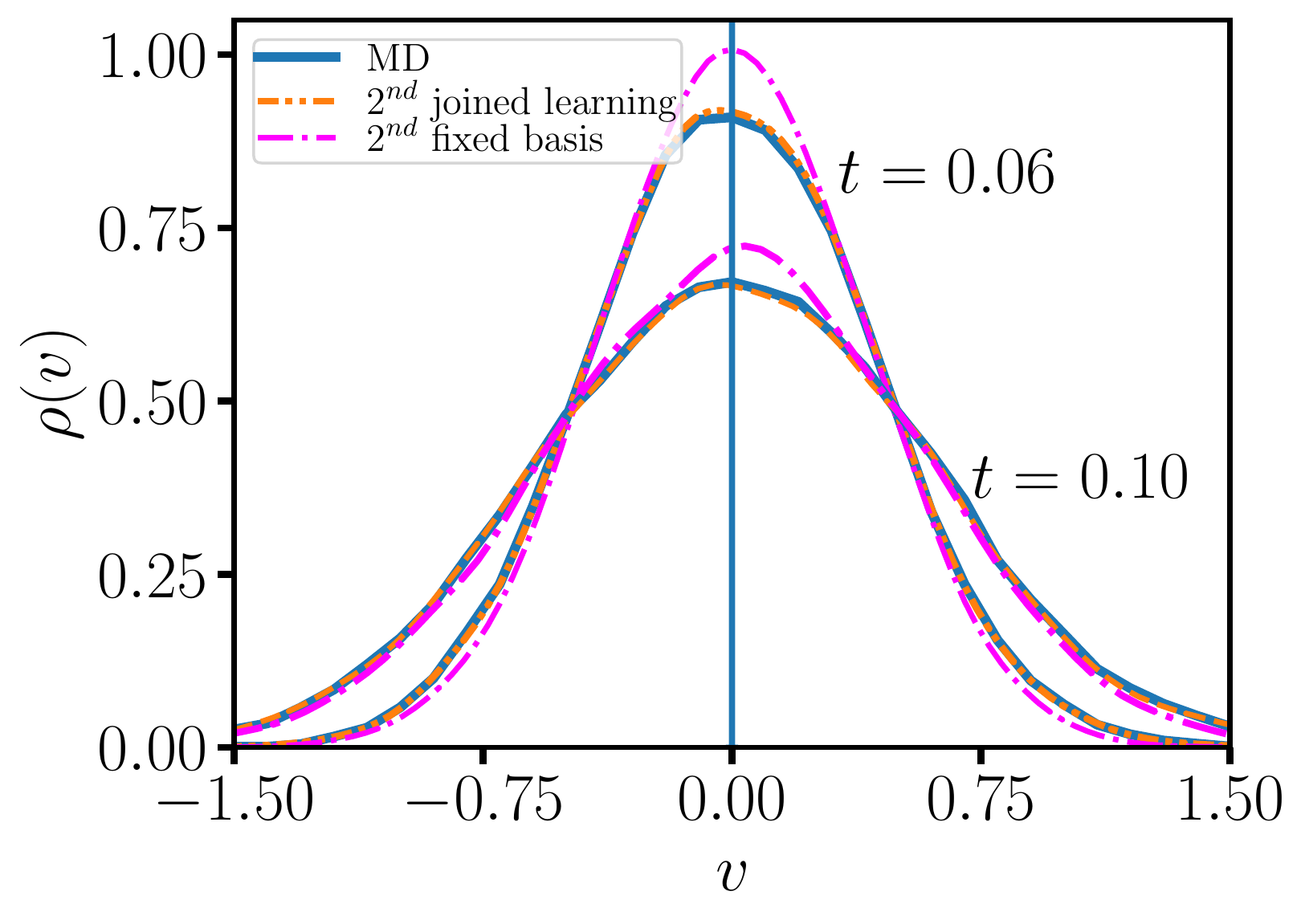}
	}
	\subfigure[]{
	\centering
		\hspace*{0em}\includegraphics[scale=0.45]{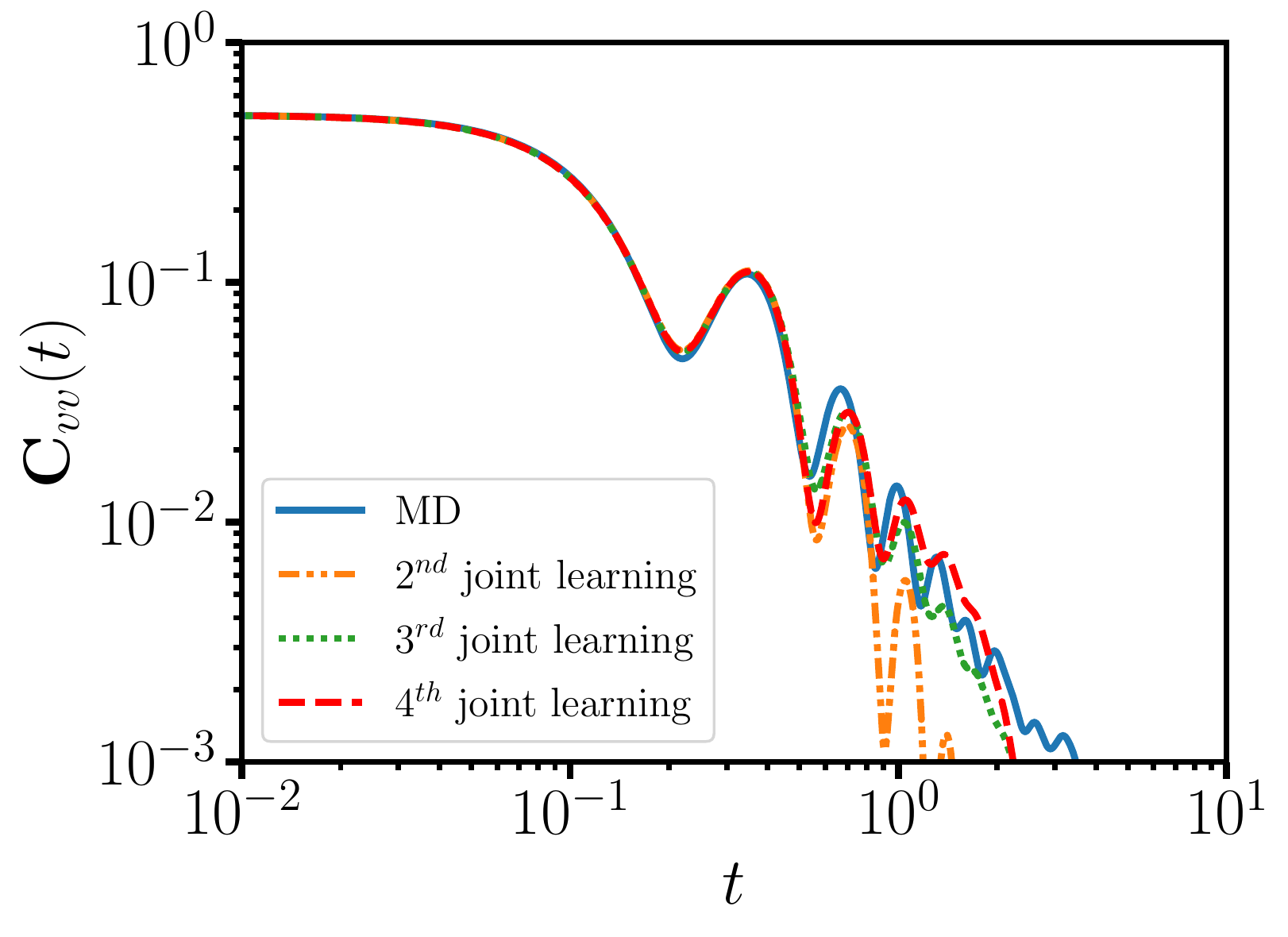}
	}
	\subfigure[]{
	\centering
		\hspace*{0em}\includegraphics[scale=0.45]{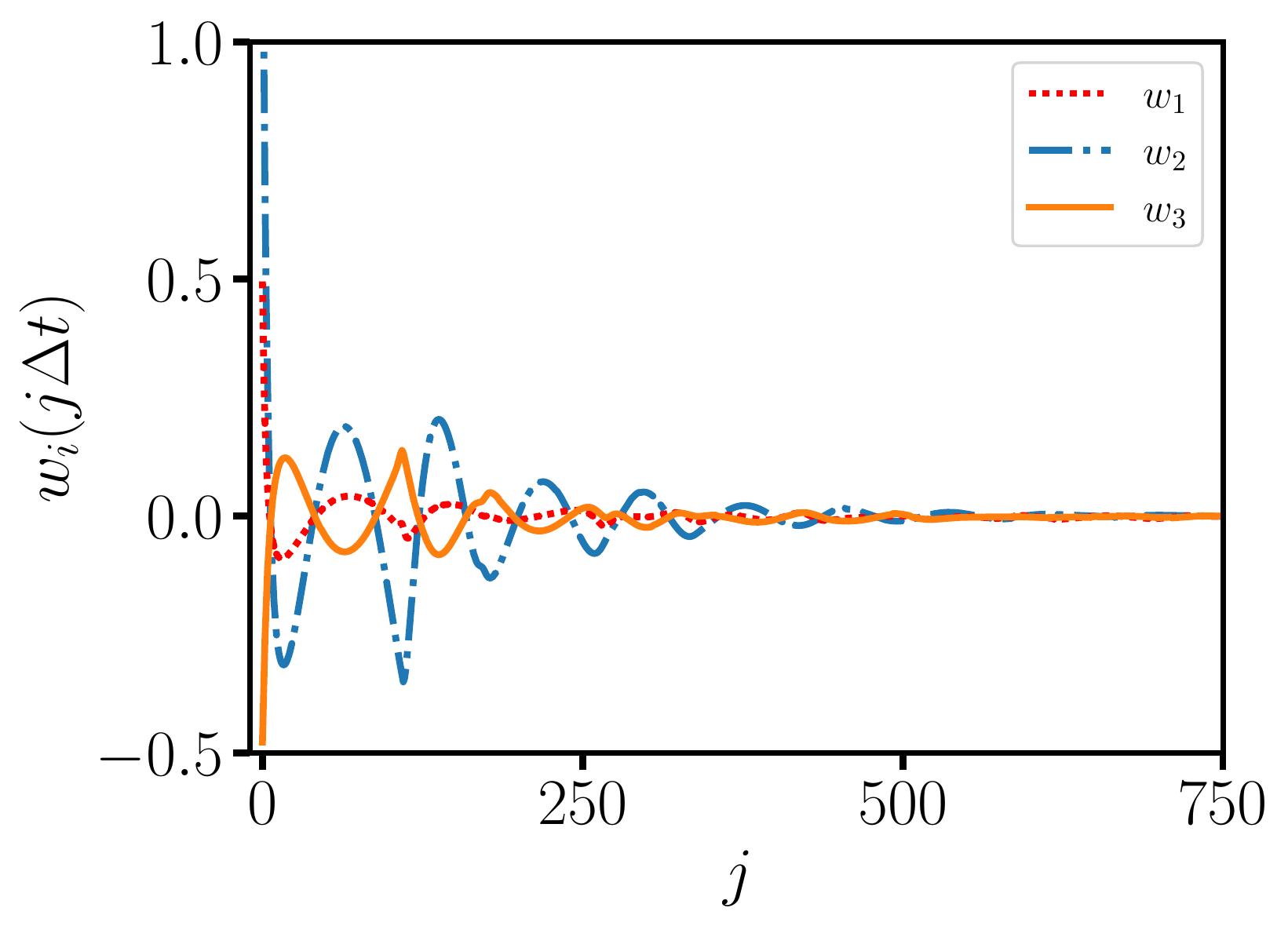}
	}
	\caption{Numerical results for a tagged particle in the solvent particle bath. (a) Velocity correlation function $C_{vv}(t)$ obtained from the full MD model and the reduced models constructed by the present method based on the joint learning approximation, the rational function approximation \cite{Lei2016}, and the Petrov-Galerkin projection with fixed bases \cite{Lei_Li_JCP_2021}. 
	(b) Predicted evolution of the probability density function of the particle velocity obtained from the full MD and the different reduced models with the second-order approximation. The initial velocity $v$ is set to $0$ (the vertical line). 
	(c)  $C_{vv}(t)$ obtained from the full MD model and different orders of the present joint learning approximation.
	(d) Encoder weights for the three non-Markovian features obtained from the present joint learning with the fourth-order approximation. }
	\label{fig:tag_particle}
\end{figure}

The reduced dynamics in terms of the tagged particle is constructed in form of Eq. \eqref{eq:SDE_simple} along with the learning of the non-Markovian features $\left\{\bm \zeta_i\right\}_{i=1}^n$. The free energy $U(\mb q)$ vanishes for this case. For comparison, we also construct the reduced model using the previous approaches based on the Petrov-Galerkin projection (named as fixed-basis) \cite{Lei_Li_JCP_2021} and the rational function approximation \cite{Lei2016}.  Fig. \ref{fig:tag_particle}(a) shows the velocity correlation function of constructed models using two non-Markovian features, or equivalently, two projection bases, as well as the second-order rational function approximation. The model constructed by the present (named as the joint-learning)  method shows the best agreement with the full model based on the molecular dynamics (MD) simulations.  The model accuracy can be further examined by the evolution of probability density function (PDF) of the particle velocity. Specifically, we fix the velocity to be zero as $t = 0$ and sample the instantaneous PDF thereafter. Fig. \ref{fig:tag_particle}(b) shows the obtained PDF at $t = 0.06$. The present approach yields more accurate result than the Petrov-Galerkin method. As shown in Fig. \ref{fig:tag_particle}(c), the accuracy of the reduced model shows further improvement as we increase the number of non-Markovian features. In particular, the reduced model with the fourth-order approximation can accurately capture the oscillations over the full regime. Fig. \ref{fig:tag_particle}(d) shows the obtained encoder weights of the fourth-order approximation. All of the three encoder functions show pronounced oscillations near $t=0$ and decay to $0$ for large $t$. Unlike the empirically chosen fractional-derivative bases in Ref. \cite{Lei_Li_JCP_2021}, the present method enables the encoders to be optimized for the best approximation of the non-local statistics, and therefore yields more accurate results. 

\begin{figure*}[htbp]
\center
\includegraphics[scale=0.3]{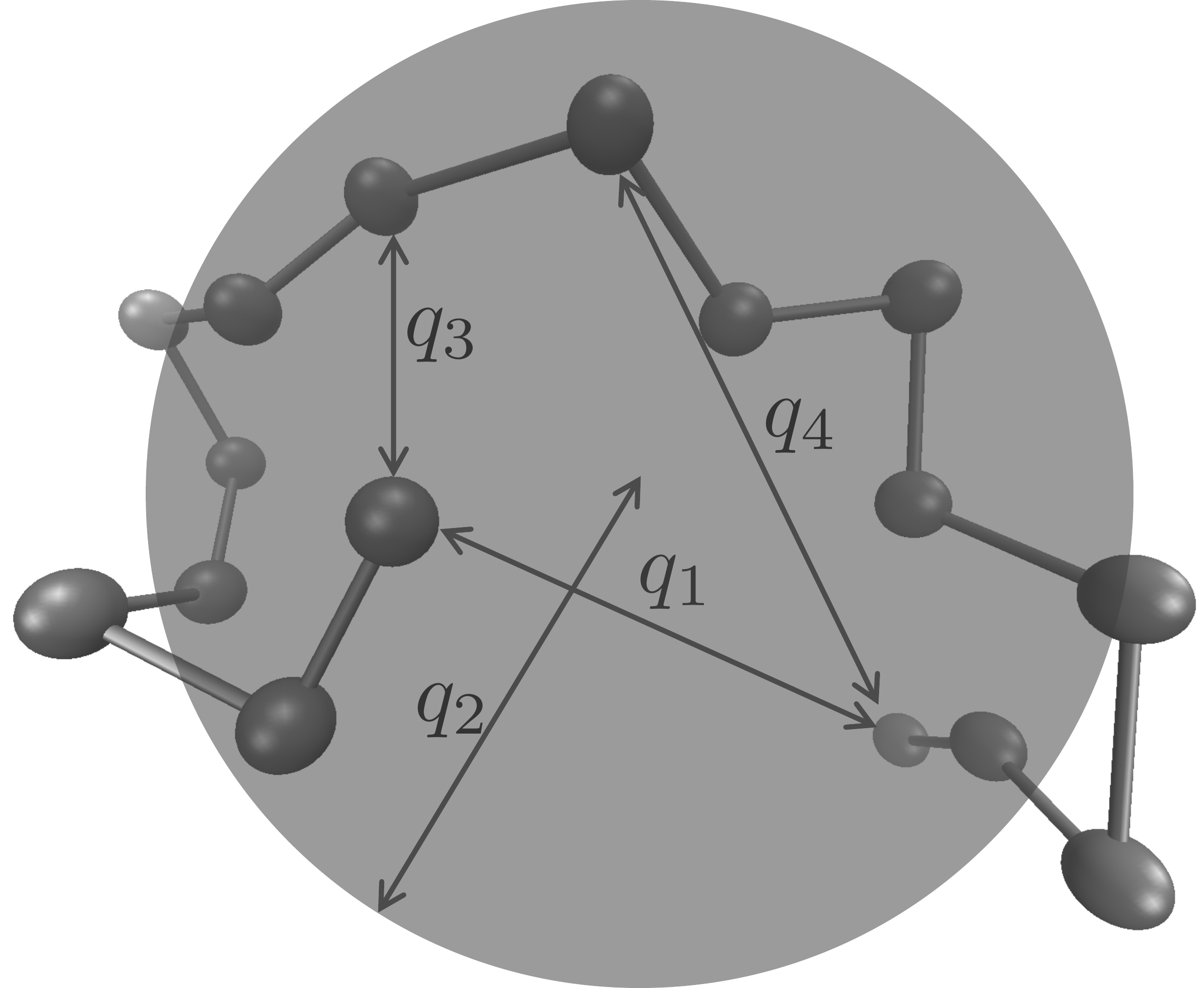}
\caption{A sketch of a chain-molecule represented by united atom model. Reduced models are constructed with respect to a four-dimensional resolved vector $\mb q$, which represents the end-to-end distance ($q_1$), the radius of gyration ($q_2$), and the end-to-middle distances ($q_3$ and $q_4$), respectively.}
\label{fig:molecule_sketch}
\end{figure*}

\subsection{One-dimensional reduced model of a polymer molecule}
We consider the reduced dynamics of a polymer molecule consisting of $N=16$ atoms. The intramolecular potential is governed by 
\begin{equation}
    V_{\rm mol}(\mb Q) = \sum_{i\neq j}^{N}V_{\rm p}(Q_{ij})+\sum_{i=1}^{N_b}V_{\rm b}(l_i)+\sum_{i=1}^{N_a}V_{\rm a}(\theta_i)+\sum_{i=1}^{N_d}V_{\rm d}(\phi_i),
\end{equation}
where $l_i$, $\theta_i$, $\phi_i$ represent the individual bond length, bond angle, and dihedral angle, respectively.
$V_{\rm p}$,  $V_{\rm b}$, $V_{\rm a}$, and $V_{\rm d}$ represent the pairwise Lennard-Jones, finite extensible nonlinear elastic bond, harmonic angle, and multiharmonic dihedral interactions whose explicit forms are specified in Appendix \ref{app:MD_polymer}. The atom mass is set to unit, and thermal temperature $k_BT$ is set to $1.0$. Fig. \ref{fig:molecule_sketch} shows a sketch of the polymer molecule. 

To examine the effectiveness of the present method, we first construct a 1D reduced dynamics in terms of the end-to-end distance $q_1 = \left\Vert \mb Q_1 - \mb Q_N\right\Vert$ as done in the previous work \cite{Lei_Li_JCP_2021} based on the Petrov-Galerkin method, and compare the numerical results obtained from the two methods.   
Figure \ref{fig:1D_molecule_system}(a) shows the velocity correlation function $C_{vv}(t) = \left\langle v_1(t) v_1(0)\right\rangle$ obtained from the full MD and different orders of fixed-basis and joint-learning approximations. With the same order of approximation, the current method yields better agreement with the MD results. Specifically, the second-order model of the current method  can capture the pattern around $t=4$ and the fourth-order model can capture the patterns around $t=0.4$ and $t = 4$. However, the previous method with the same order approximation shows limitation to accurately capture these two patterns. 

\begin{figure}
\centering
	\subfigure[]{
	\centering
		\hspace*{0em}\includegraphics[scale=0.45]{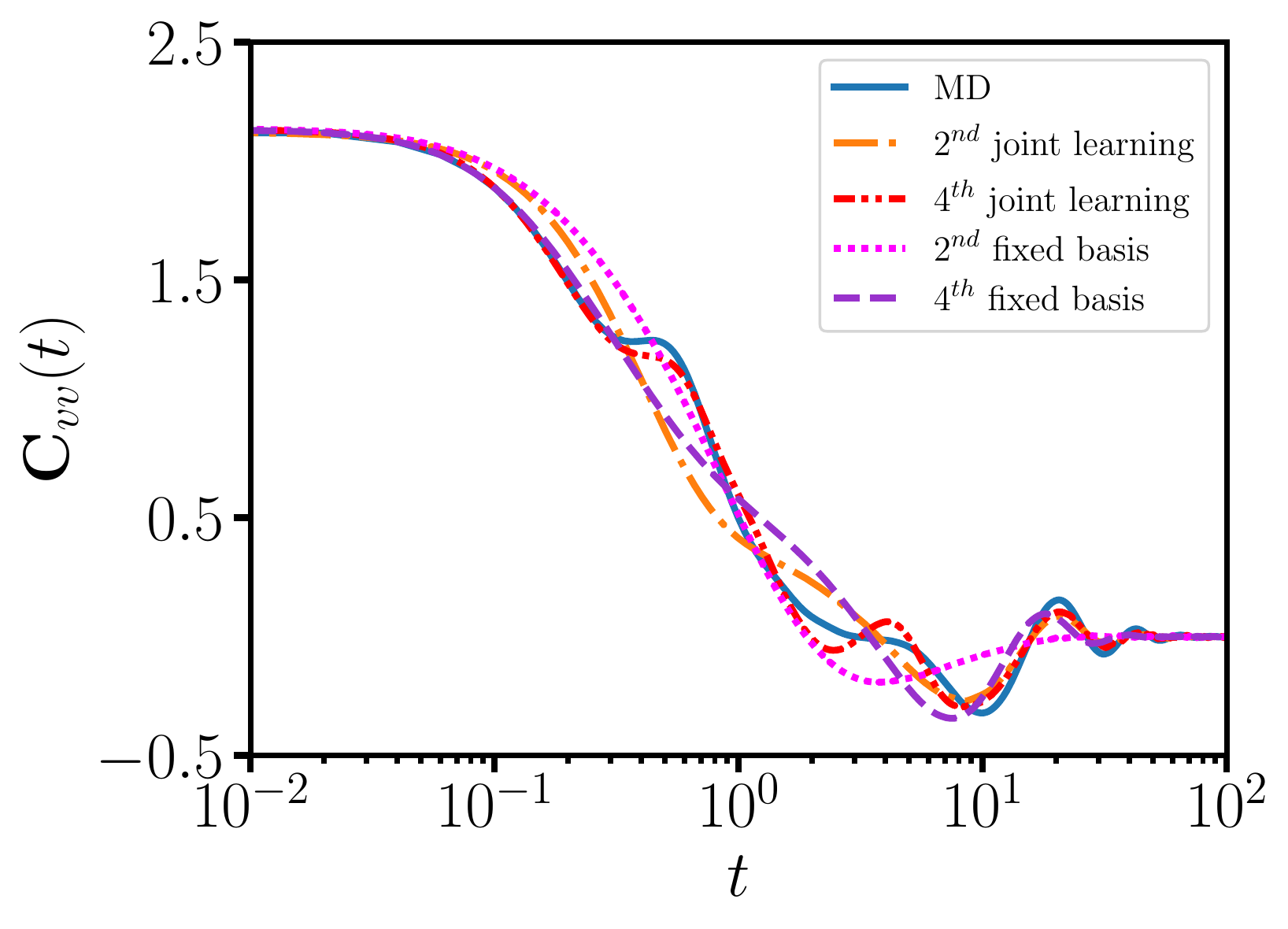}
	}
	\subfigure[]{
	\centering
		\hspace*{0em}\includegraphics[scale=0.45]{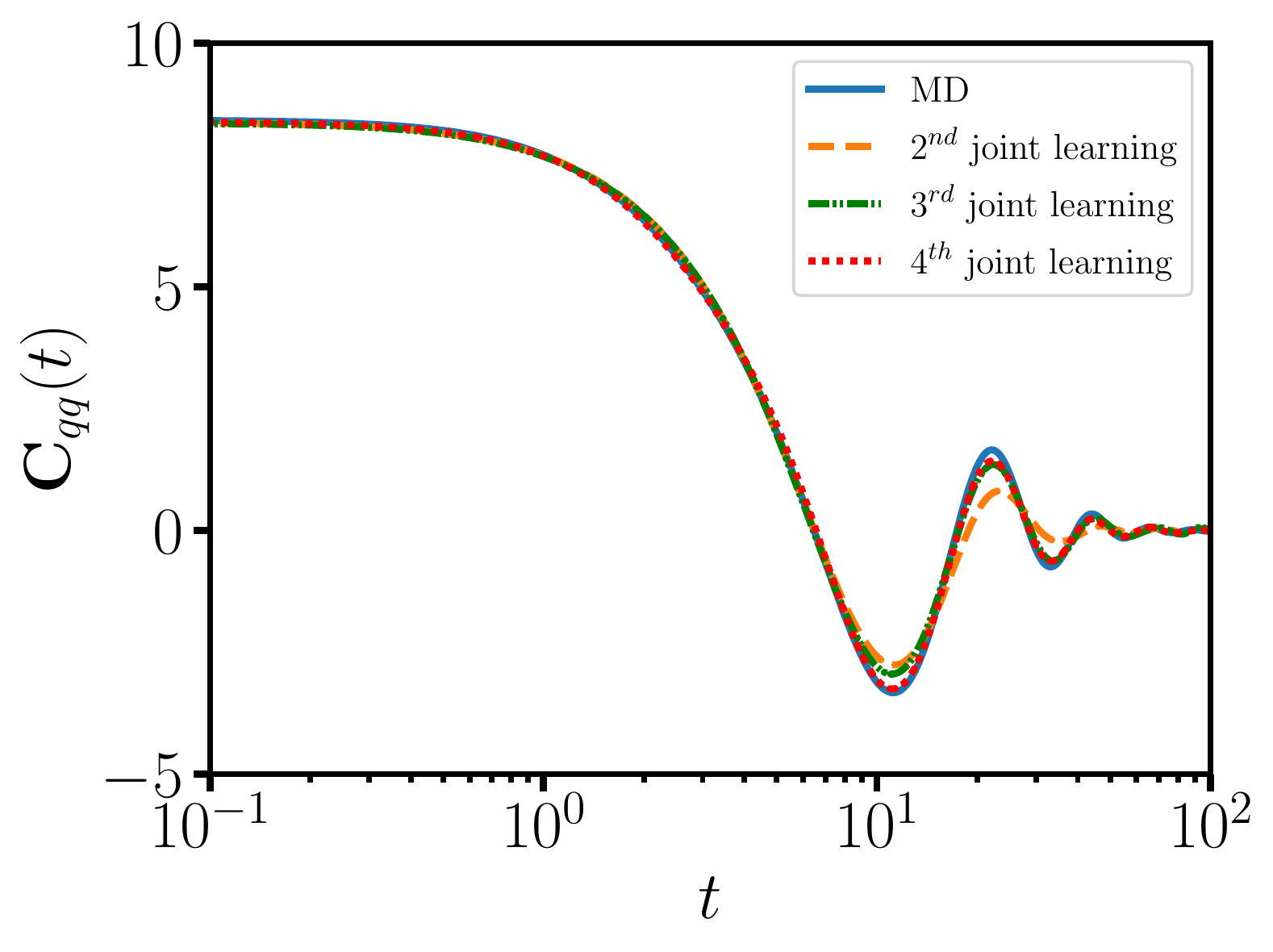}
	}
	\subfigure[]{
	\centering
		\hspace*{0em}\includegraphics[scale=0.45]{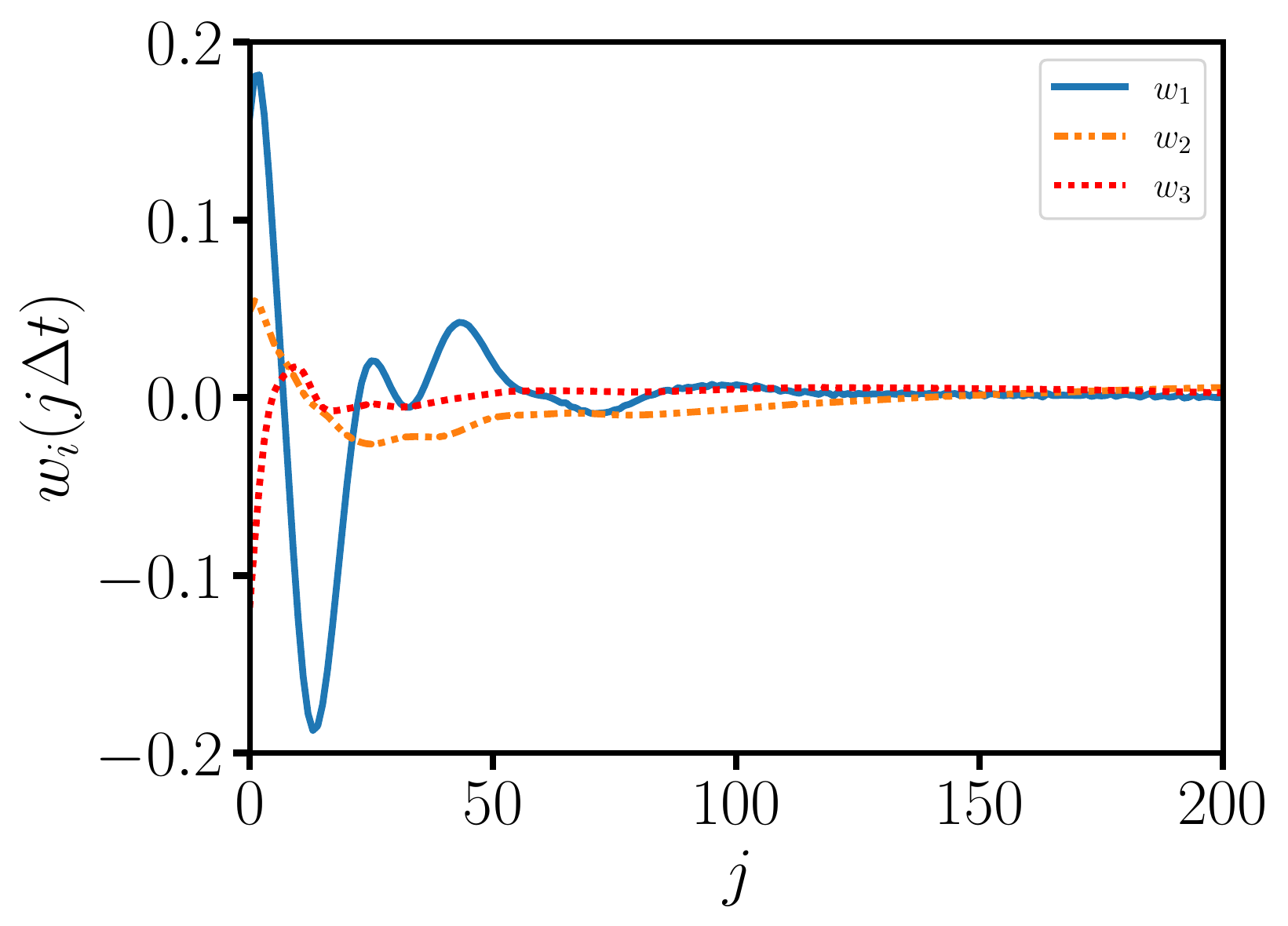}
	}
	\subfigure[]{
	\centering
		\hspace*{0em}\includegraphics[scale=0.45]{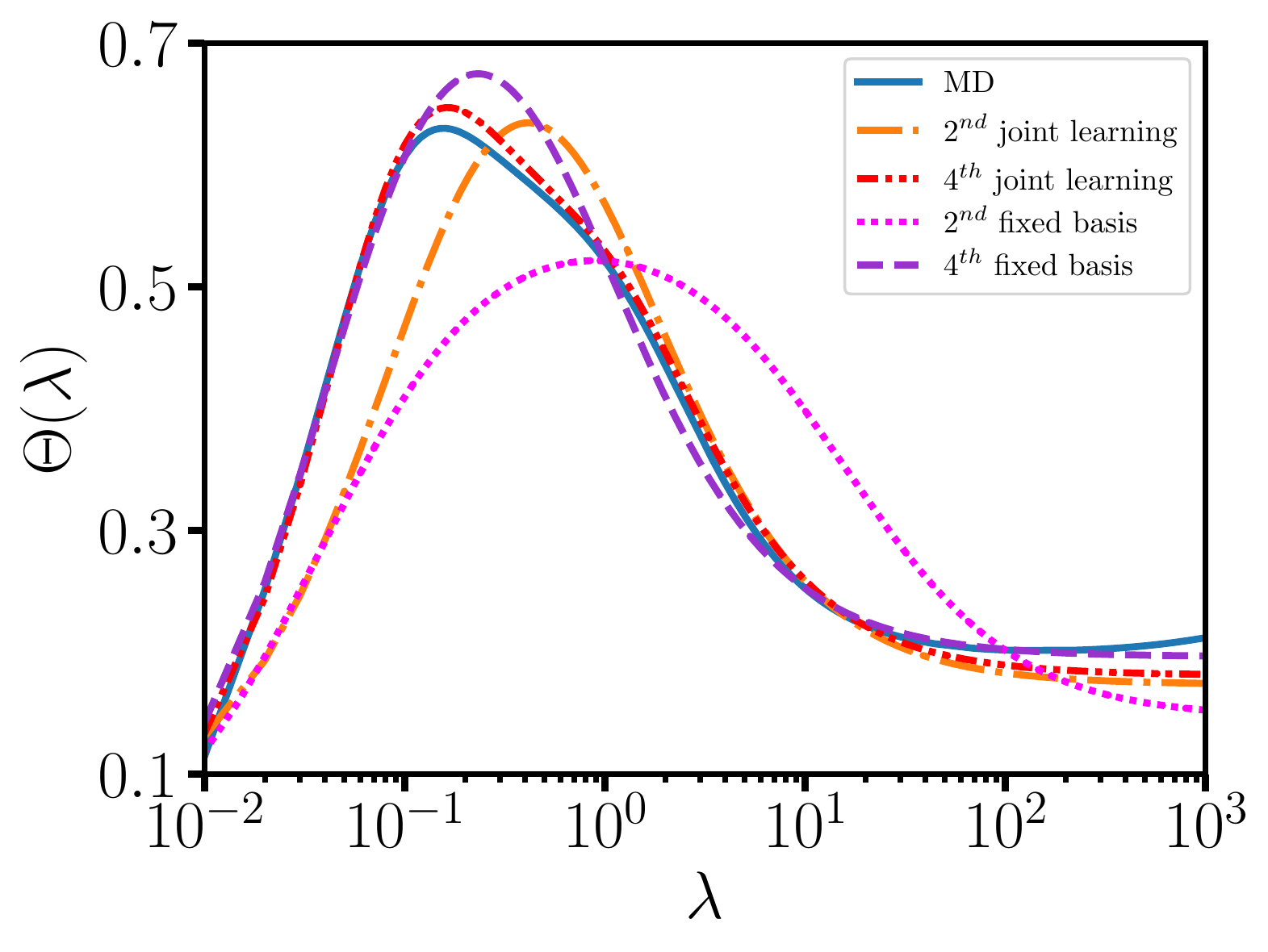}
	}
	\caption{Numerical results of a one-dimensional reduced model representing the dynamics of the end–end distance of a polymer molecule system. (a-b) Velocity correlation function $C_{vv}(t)$ and the Laplace transform of the memory function $\Theta(\lambda)$ obtained from the full MD simulations and the different orders of the present joint learning approximation, and the Petrov-Galerkin projection with fixed basis approximation. (c) Displacement correlation function $C_{qq}(t)$ obtain form the full MD and different orders of the joint learning approximation. (d) Encoder weights for the three non-Markovian features of the reduced model with the fourth-order approximation.}
	\label{fig:1D_molecule_system}
\end{figure}

Figure \ref{fig:1D_molecule_system}(b) shows the displacement auto-correlation function 
$C_{qq}(t) = \left\langle q_1(t) q_1(0)\right\rangle$
obtained from full MD and the reduced models constructed by the present method with different number of non-Markovian features. As we introduce more features, the predicted correlation functions approaches the MD results. In particular, the fourth-order model can capture the oscillations of the MD results at $t = 10$ and $t= 25$. Figure \ref{fig:1D_molecule_system}(c) shows the encoder weights of non-Markovian features for the fourth-order approximation. Similar to the tagged particle system, the encoder functions exhibit pronounced oscillations at the short time and decay to zero at longer time.

The accuracy of the constructed reduced models can be further examined by comparing the embedding memory kernels $\tilde{\bm\theta}(t)$ with the full MD model. Figure \ref{fig:1D_molecule_system}(d) shows the Laplace transform of the memory kernel of the reduced models $\tilde{\bm \Theta}(\lambda) = \int_0^{+\infty} \tilde{\bm\theta}(t) \exp{(-t/\lambda)} \diff t$. The MD kernel $\bm \Theta(\lambda)$ is obtained by $\bm \Theta(\lambda) = -\mb G(\lambda) \mb H(\lambda)^{-1}$, where $\mb G(\lambda)$ and $\mb H(\lambda)$ are the Laplace transform of the correlation matrices $\mb g(t) = \left\langle \mb M \dot{\mb v}(t) + \nabla U(\mb q), \mb q(0)\right\rangle$ and $\mb h(t) = \left\langle \mb v(t), \mb q(0)\right\rangle$. Compared with the previous method, the current method yields better agreement with MD results. Specifically, the second$\mhyphen$ and fourth-order of the joint learning approximation, and the fourth$\mhyphen$order of the fixed basis approximation show good agreement with the MD result $\bm \Theta(\lambda)$ for $\lambda$ between $1$ and $1000$. Furthermore, the fourth-order model of the joint learning approximation can further capture the  pronounced peak regime of the MD results near $\lambda=0.1$.  We emphasize that the \emph{memory kernel $\tilde{\bm\theta}(t)$ is not explicitly constructed during the learning process}; $\tilde{\bm\theta}(t)$ approaches $\bm\theta(t)$ as we impose the constraint \eqref{eq:correlation_evolution} such that the correlation matrices of the reduced dynamics match the ones of the full model. This enables us to circumvent the direct fitting of the matrix-valued memory function for multi-dimensional GLEs, and efficiently construct the numerically-stable reduced model that retains the non-local statistics and coherent noise as shown in the following example.

\subsection{Four-dimensional reduced model of a polymer molecule}
Finally, we construct a reduced model in terms of a four-dimensional resolved vector $\mb q = [q_1, q_2, q_3, q_4]$ defined by 
\begin{equation}
\begin{split}
q_1 &= \left\Vert \mb Q_1 - \mb Q_N\right\Vert, \\
q_2^2 &= \frac{1}{N}\sum_{i=1}^{N}\Vert\mathbf{Q}_i-\mathbf{Q}_c\Vert^2,  \quad \mathbf{Q}_c= \frac{1}{N}\sum_{i=1}^{N}\mathbf{Q}_i, \\
q_3 &=\left\Vert \mathbf{Q}_{\lfloor{\frac{N}{2}}\rfloor}-\mathbf{Q}_1\right\Vert, \\
q_4 &=\left\Vert \mathbf{Q}_{\lceil{\frac{N}{2}}\rceil}-\mathbf{Q}_N\right\Vert,
\end{split}
\label{eq:resolved_variables}
\end{equation}
where $q_1$, $q_2$, $q_3$, and $q_4$ represent the end-to-end distance, radius of gyration, and two center-to-end distances, respectively. The four-dimensional free energy function $U(\mb q)$ is constructed by matching the average force sampled from the constraint molecular dynamics and represented by a neural network; we refer to Appendix \ref{app:free_energy} for details. 
Rather than constructing the four-dimensional GLE kernel $\bm\theta(t)$, we directly learn the reduced model \eqref{eq:SDE_simple} by minimizing the loss function \eqref{eq:loss_function}. 

\begin{figure}
\centering
	\subfigure[]{
	\centering
		\includegraphics[scale=0.45]{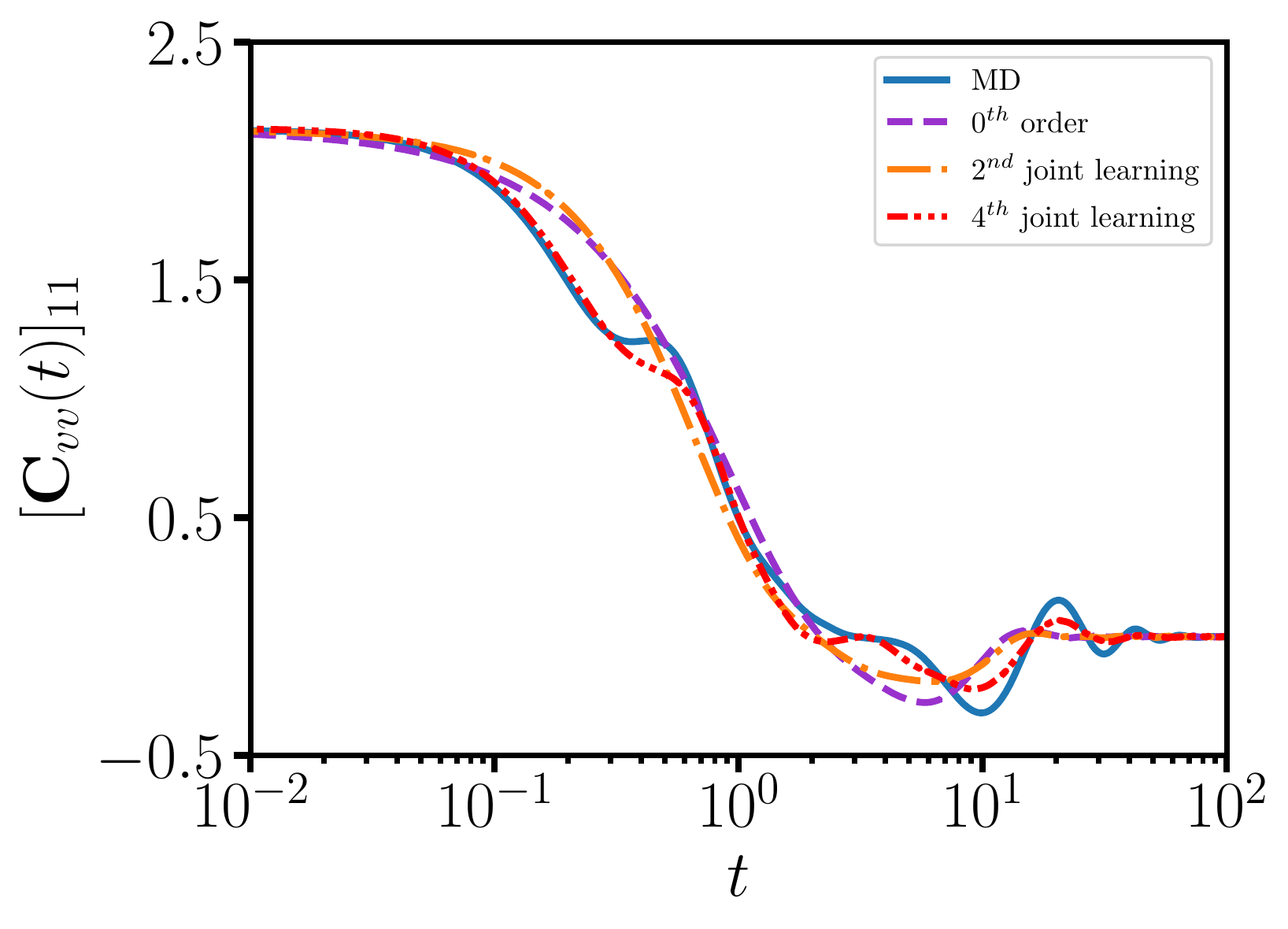}
	}
	\subfigure[]{
	\centering
		\includegraphics[scale=0.45]{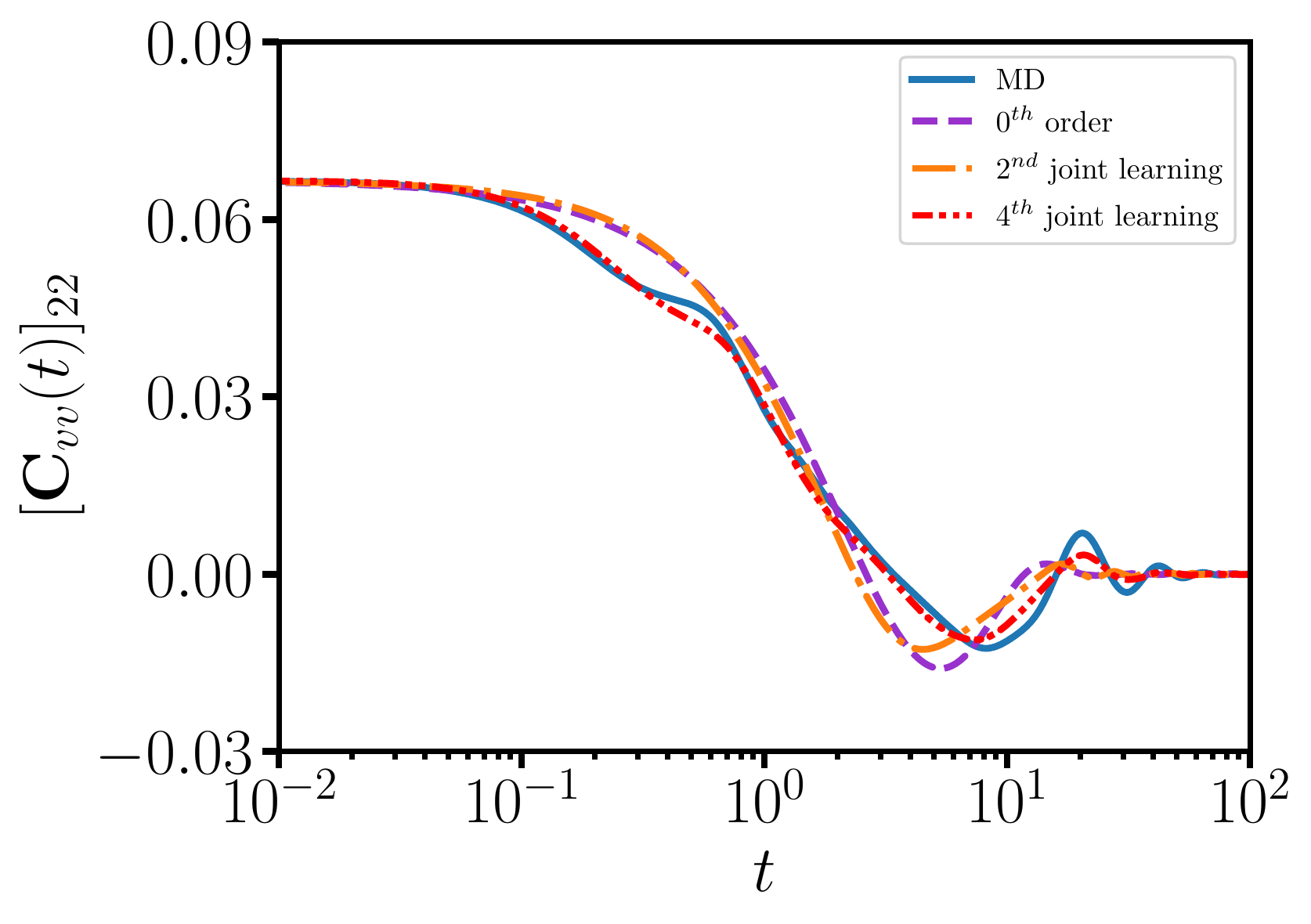}
	}
	\subfigure[]{
	\centering
		\includegraphics[scale=0.45]{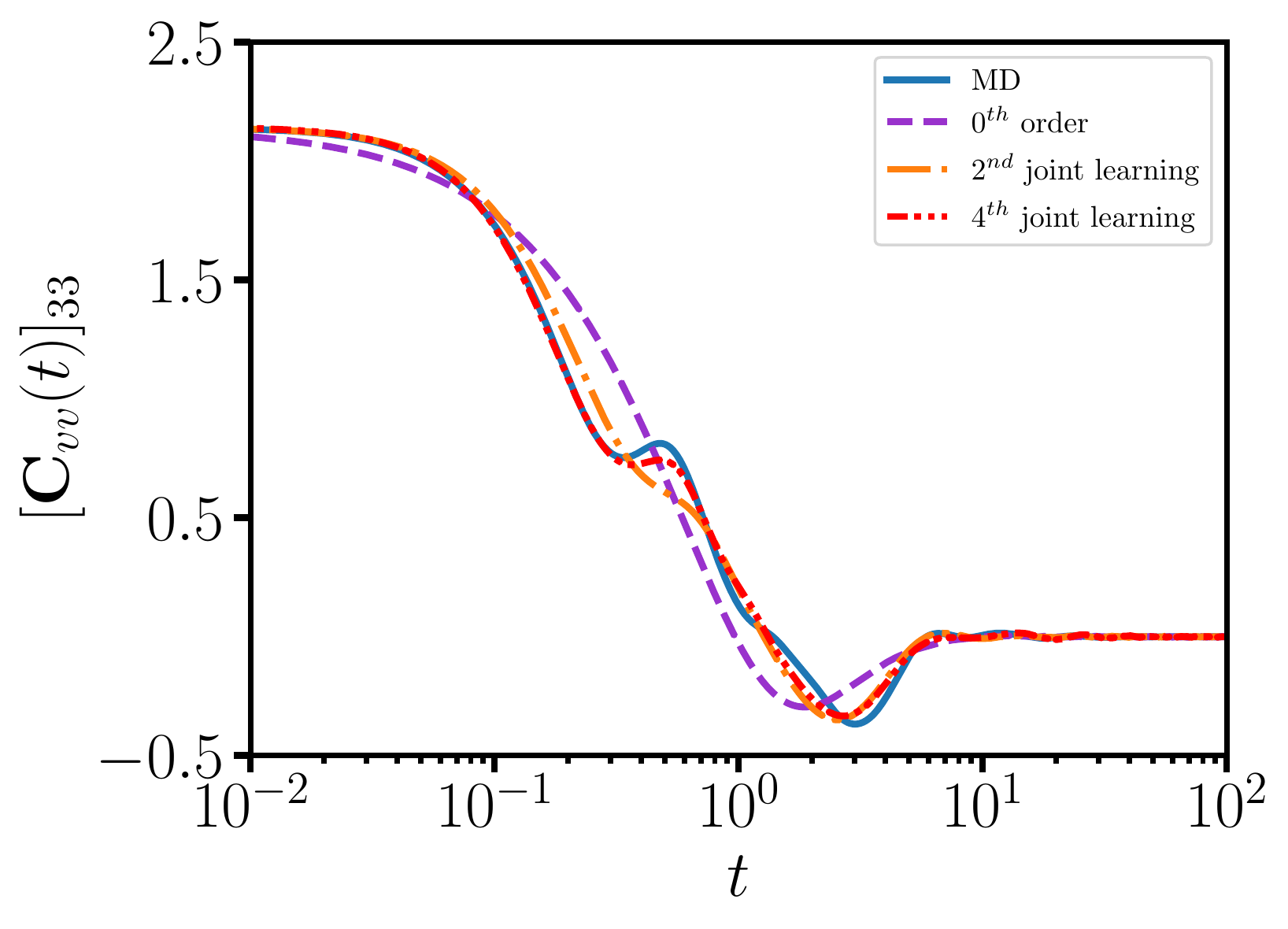}
	}
	\subfigure[]{
	\centering
		\includegraphics[scale=0.45]{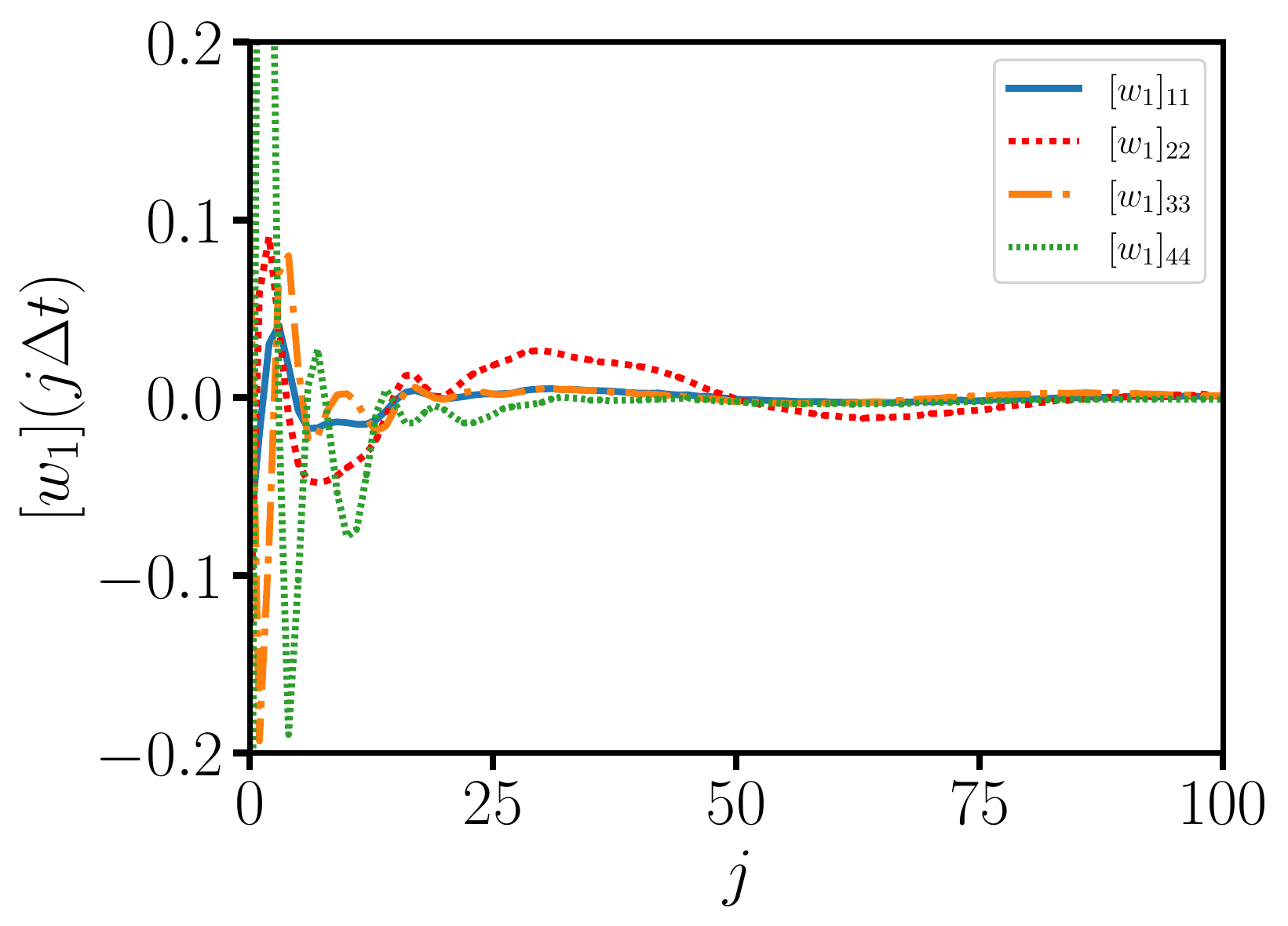}
	}
	\caption{Numerical results of a four-dimensional reduced model representing the dynamics of the of a polymer molecule system with conformation states characterized by the resolved variables $\mb q$ defined by Eq. \eqref{eq:resolved_variables}. (a-c)  Diagonal components of velocity correlation function $\mb C_{vv}(t) = \left\langle \mb v(t) \mb v(0)^T\right\rangle$. $\left[\mb C_{vv}(t)\right]_{44}$ is not shown as it is similar to $\left[\mb C_{vv}(t)\right]_{33}$. (d) Constructed encoder weights of the first non-Markovian features $\bm\zeta_1$ for the reduced model of the fourth-order approximation.}
	\label{fig:4D_molecule_system_part_1}
\end{figure}

Figure \ref{fig:4D_molecule_system_part_1}(a-c) show the diagonal components of the velocity correlation matrix $\mb C_{vv}(t) = \left\langle \mb v(t) \mb v(0)^T\right\rangle$ obtained from the full MD and the reduced models using different order approximations. Specifically, the components $\left[\mb C_{vv}(t)\right]_{11}$ and $\left[\mb C_{vv}(t)\right]_{33}$ show similar values near $t= 0$ since both $q_1$ and $q_3$ characterize the distances between the individual particles, e.g., $v_1 = (\mb Q_1 - \mb Q_N)\cdot(\mb V_1 - \mb V_N)/\left\Vert \mb Q_1 - \mb Q_N\right\Vert$. As the velocities of the two free-end particles and the middle particles are nearly uncorrelated, the variances of both $v_1$ and $v_3$ are close to $2k_BT$. On the long-time scale, $\left[\mb C_{vv}(t)\right]_{11}$ and $\left[\mb C_{vv}(t)\right]_{22}$ decay much slower than $\left[\mb C_{vv}(t)\right]_{33}$ and $\left[\mb C_{vv}(t)\right]_{44}$ and show pronounced oscillations near $t = 10$ and $t=25$. The differences can be understood as follow: Compared with the end-to-middle distances $q_3$ and $q_4$, the end-to-end distance $q_1$ and radius of gyration $q_2$ represent the global states of the molecular conformation. Based on the scaling law of the idealized Gaussian chain model \cite{degennes_scaling}, the relaxation time of $q_1$ and $q_2$ is proportional to $N^2$. Accordingly, $\left[\mb C_{vv}(t)\right]_{11}$ decays four times slower than $\left[\mb C_{vv}(t)\right]_{33}$, which is qualitatively consistent with the present numerical results.  

The transient dynamics of the correlation functions can be accurately captured by the reduced model. As we increase the number of non-Markovian features, the predictions show better agreement with MD results. Specifically, the zeroth-order (i.e., Langevin) model is insufficient to capture the patterns around $0.5$ and $5$. The second-order model yields an accurate prediction for $\left[\mb C_{vv}(t)\right]_{33}$ but less accurate predictions for $\left[\mb C_{vv}(t)\right]_{11}$ and $\left[\mb C_{vv}(t)\right]_{22}$. The fourth-order model yields good agreement for all the components over the full regime. Fig. \ref{fig:4D_molecule_system_part_1}(d) shows the encoder weights of the first non-Markovian feature $\bm\zeta_1$, which naturally encode the non-local statistics among the resolved variables, and decay to $0$ at large time.

\begin{figure}
\centering
	\subfigure[]{
	\centering
		\hspace*{0em}\includegraphics[scale=0.45]{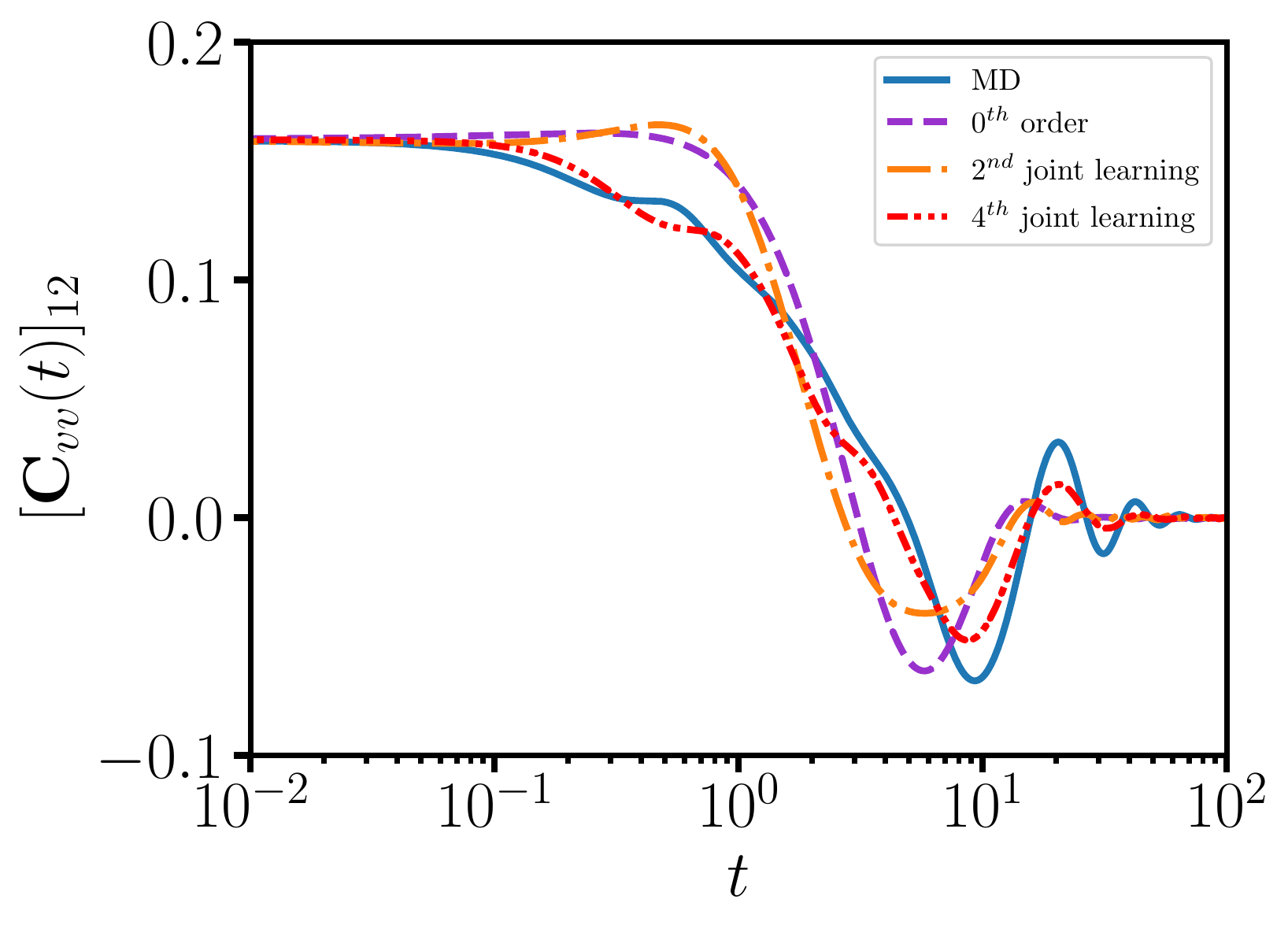}
	}
	\hfill
	\subfigure[]{
	\centering
		\hspace*{0em}\includegraphics[scale=0.45]{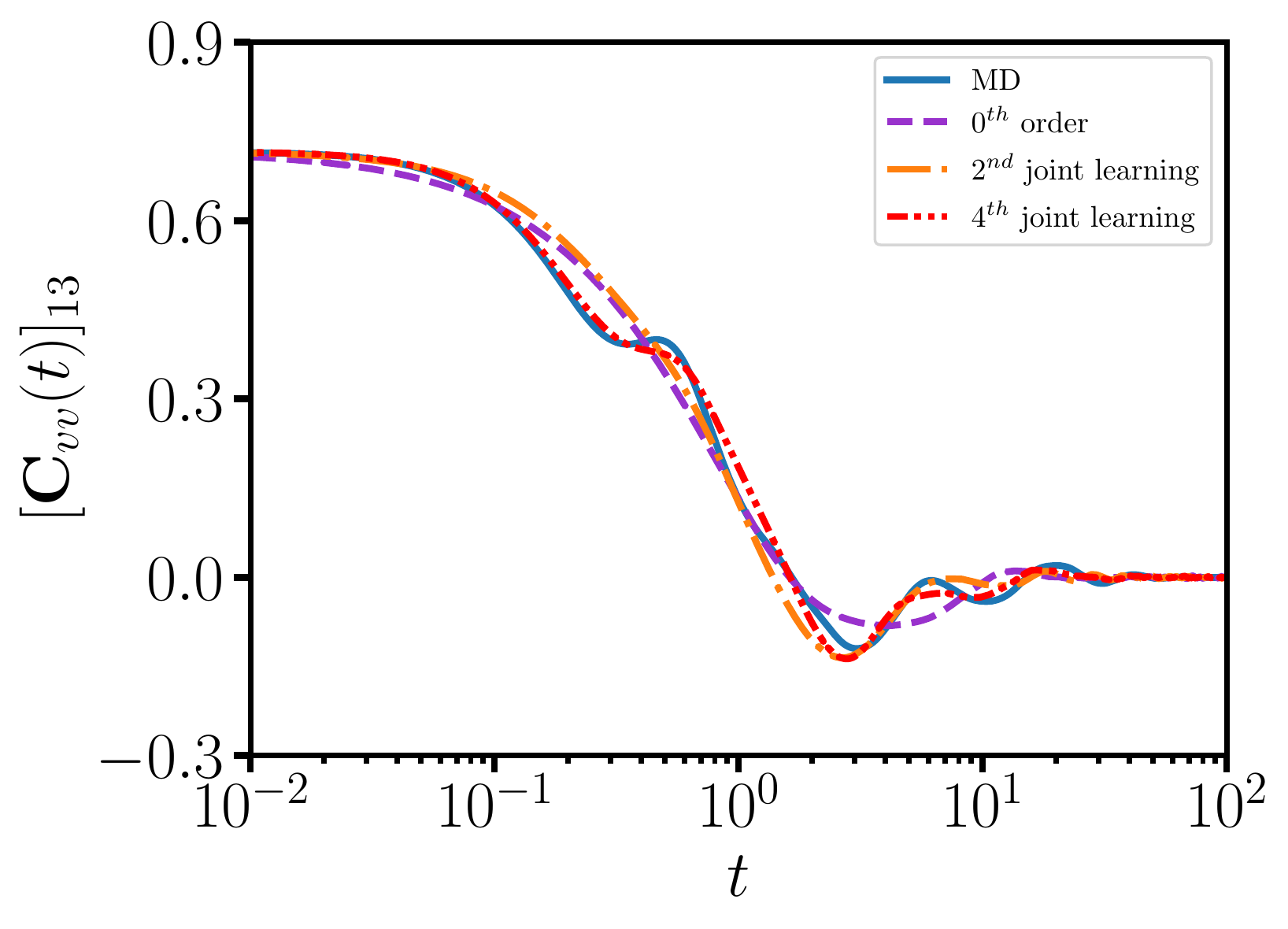}
	}
	\hfill
	\subfigure[]{
	\centering
		\hspace*{0em}\includegraphics[scale=0.45]{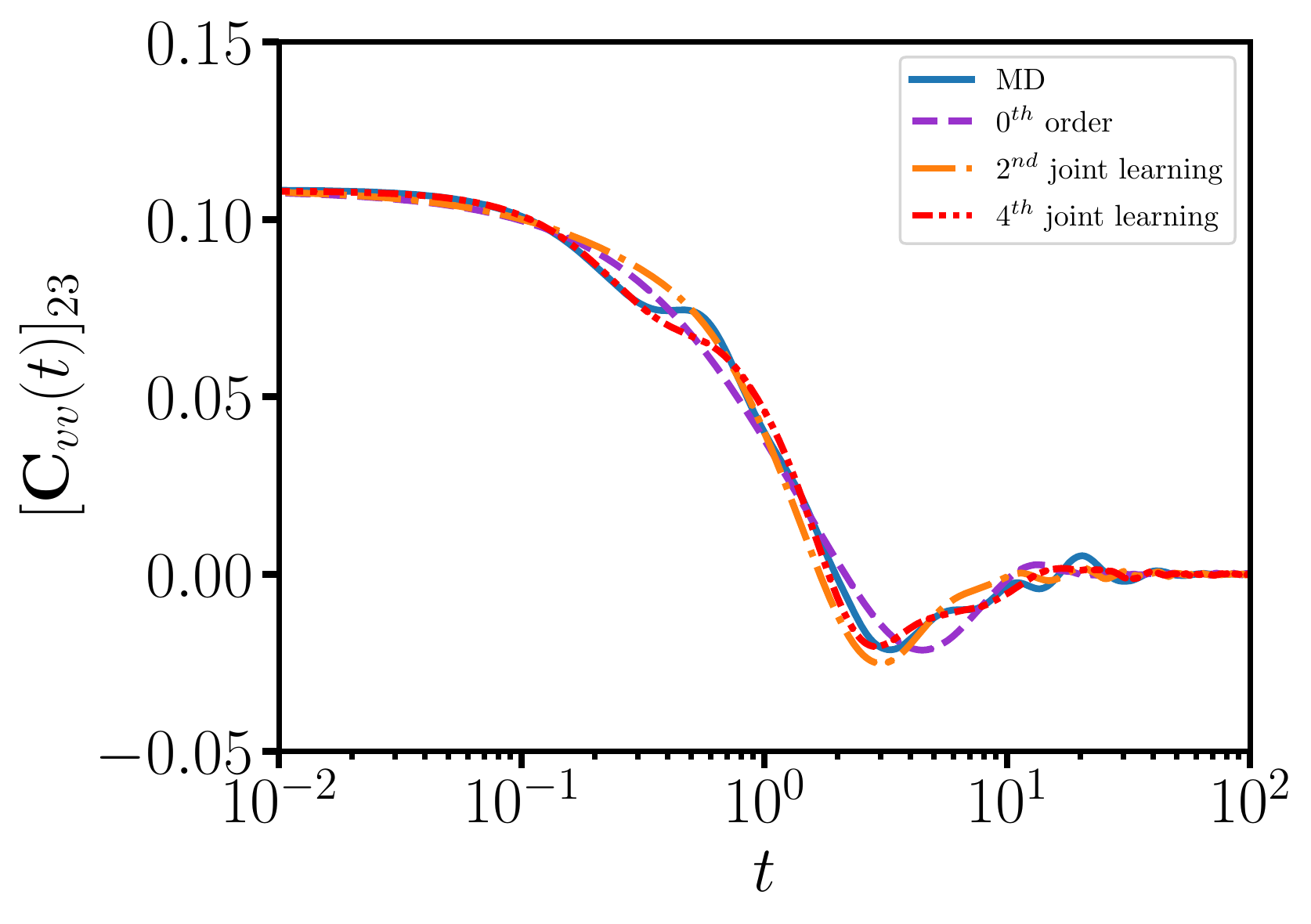}
	}
	\hfill
	\subfigure[]{
	\centering
		\hspace*{0em}\includegraphics[scale=0.45]{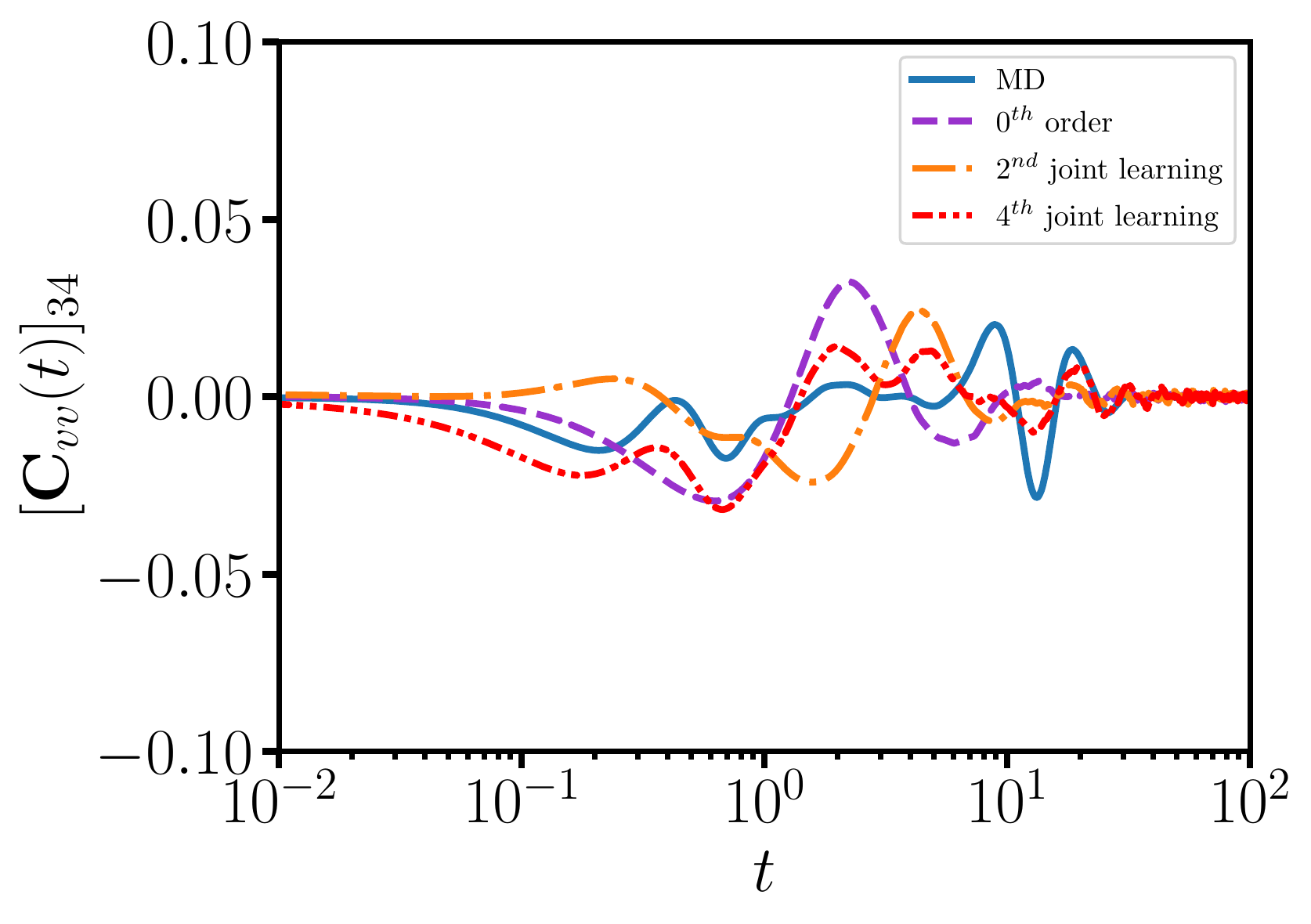}
	}
	\caption{Off-diagonal components of the velocity correlation function $\mb C_{vv}(t)$ for a polymer molecule system whose conformation states are characterized by a four-dimensional resolved vector $\mb q$ defined by Eq. \eqref{eq:resolved_variables}.
	}
	\label{fig:4D_molecule_system_part_2}
\end{figure}

Fig. \ref{fig:4D_molecule_system_part_2} shows the off-diagonal components of the velocity correlation matrix $\mb C_{vv}(t)$. Similar to the diagonal components, $\left[\mb C_{vv}(t)\right]_{12}$ represents the coupling between the dynamics of two global conformation states and therefore exhibits the longest correlation with pronounced oscillations at $t=10$ and $t=25$. On the other hand, $\left[\mb C_{vv}(t)\right]_{13}$ and $\left[\mb C_{vv}(t)\right]_{23}$ represent the coupling between a global state and semi-global state, and therefore exhibit intermediate correlation. In addition,  $\left[\mb C_{vv}(t)\right]_{34}$ exhibits weaker correlation compared with the other components since the coupling between the dynamics of $q_3$ and $q_4$ is mainly governed by the local bond$\mhyphen$ and angle$\mhyphen$interactions associated with $8\mhyphen$th and $9\mhyphen$th atom. The predictions of the second-order reduced model show fairly good agreement with the full MD results for $\left[\mb C_{vv}(t)\right]_{13}$ and $\left[\mb C_{vv}(t)\right]_{23}$ but less agreement for $\left[\mb C_{vv}(t)\right]_{12}$. The fourth-order reduced model yields good agreement for all the components.

Fig. \ref{fig:4D_memory_kernel} shows the components of the embedded matrix-valued kernels  in the Laplace space obtained from the full MD and the reduced models.  In particular,  $\tilde{\bm\Theta}(\lambda)$   obtained from the second-order model shows good agreement with $\bm\Theta(\lambda)$ obtained from the full MD within the regime of large $\lambda$. The fourth-order model yields good agreement over the full regime, which is consistent with the accurate prediction of the velocity correlation functions shown in Fig. \ref{fig:4D_molecule_system_part_1} and \ref{fig:4D_molecule_system_part_2}. While the kernel function $\bm\theta(t)$ is not explicitly constructed in the present method, the accurate recovery of $\bm\Theta(\lambda)$ verifies that the constructed models faithfully retain the non-Markovian dynamics of the resolved variables. 

\begin{figure}
\centering
	\subfigure[]{
	\centering
		\hspace*{0em}\includegraphics[scale=0.45]{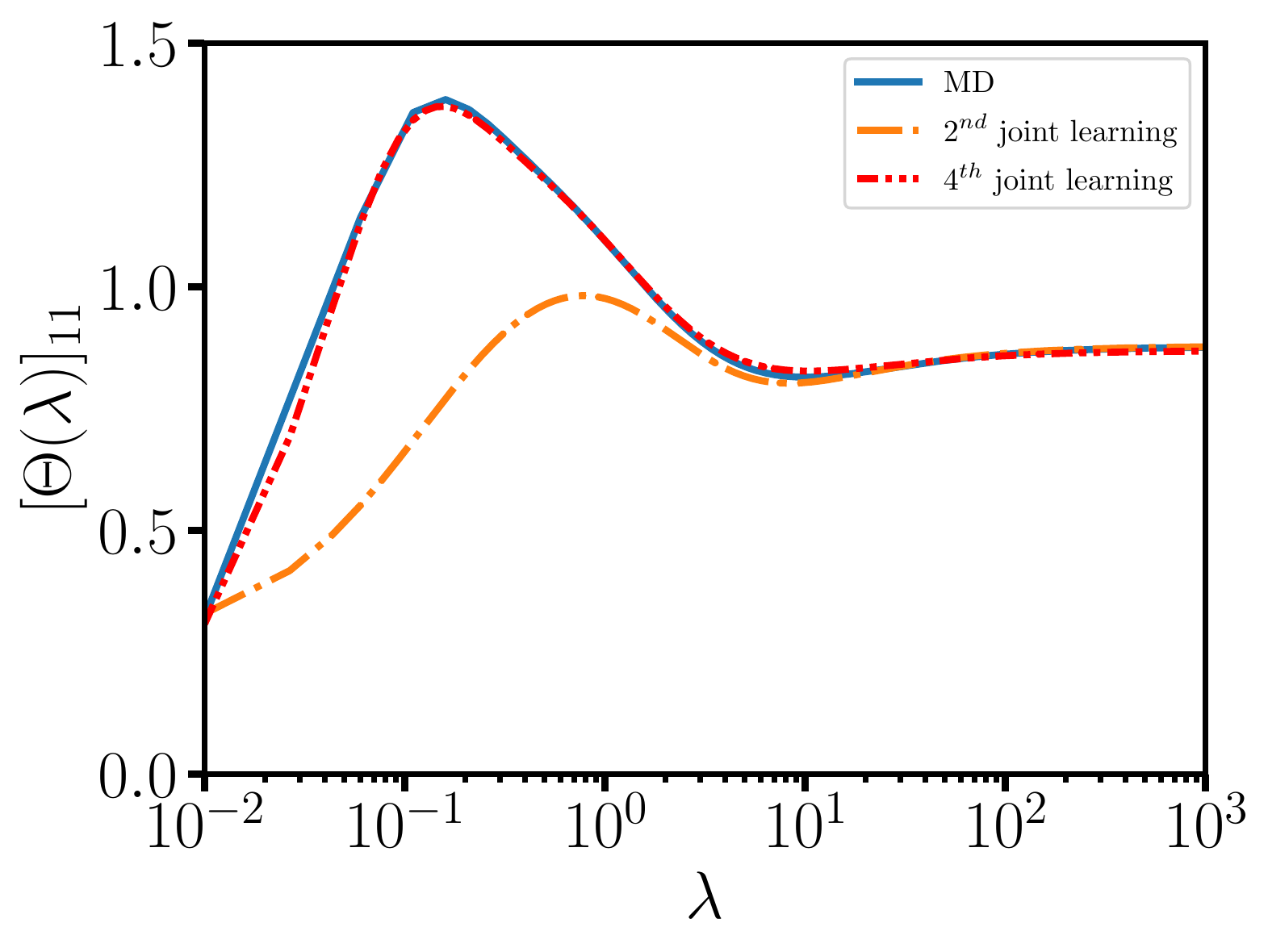}
	}
	\subfigure[]{
	\centering
		\hspace*{0em}\includegraphics[scale=0.45]{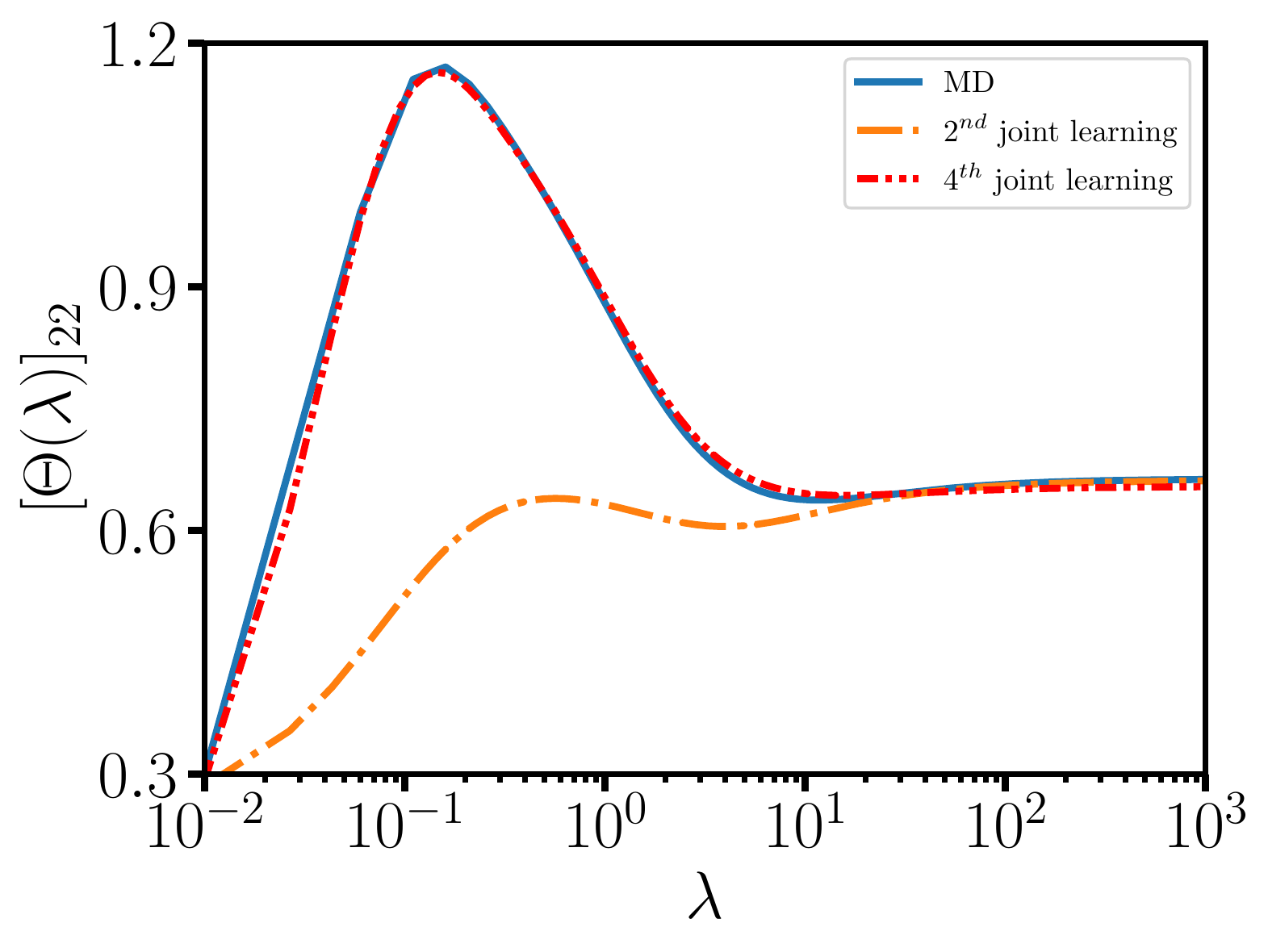}
	}
	\subfigure[]{
	\centering
		\hspace*{0em}\includegraphics[scale=0.45]{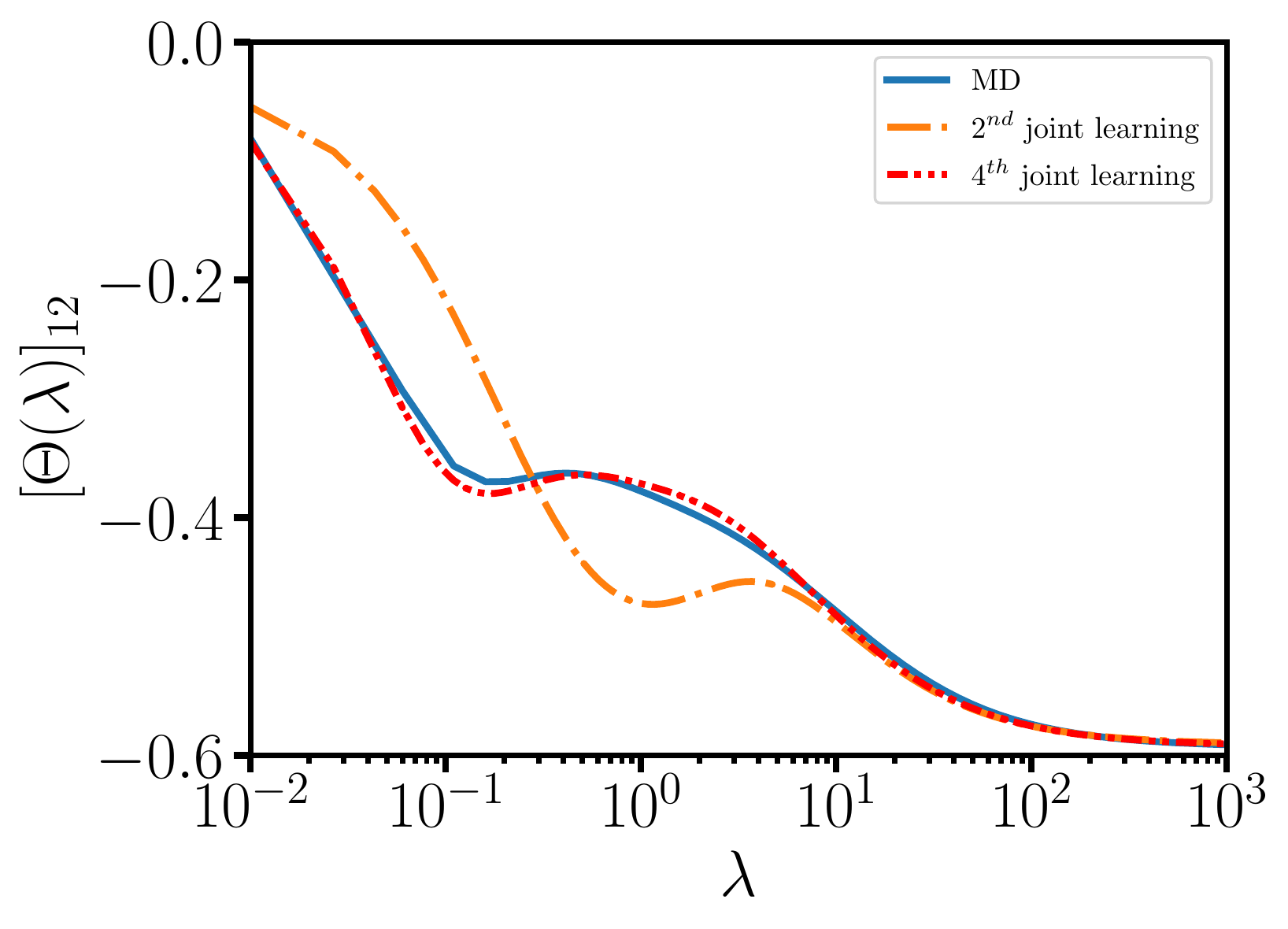}
	}
	\subfigure[]{
	\centering
		\hspace*{0em}\includegraphics[scale=0.45]{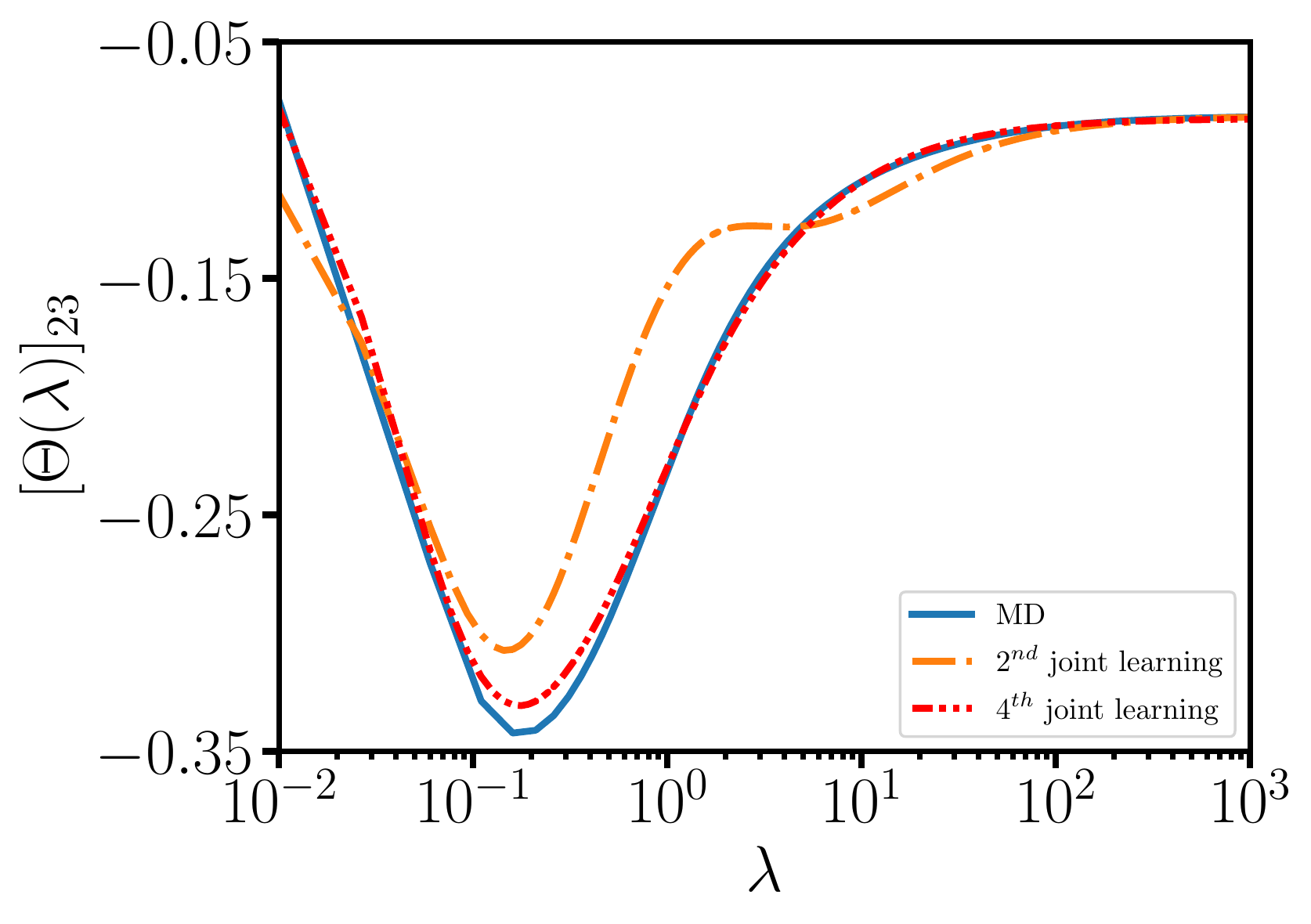}
	}
	\caption{Components of the embedded matrix-valued kernel $\bm\Theta(\lambda)$ in the Laplace space obtained from the full MD and a four-dimensional reduced model of a polymer molecule system.}
	\label{fig:4D_memory_kernel}
\end{figure}

\section{Discussion}
\label{sec:discussion}
In this study, we developed a data-driven approach to accurately learn the stochastic reduced dynamics of full Hamiltonian systems with non-Markovian memory. The method essentially provides an efficient approach to approximate the multi-dimensional generalized Langevin equation. Rather than directly fitting the matrix-valued memory kernel, the present method seeks a set of non-Markovian features whose evolution naturally encodes with the orthogonal dynamics of the resolved variables, and establishes a joint learning of the extended dynamics in terms of both the resolved variables and the non-Markovian features. Compared with the previous studies based on the rational function approximation \cite{Lei2016} and 
the Petrov-Galerkin projection \cite{Lei_Li_JCP_2021} with the pre-selected fractional derivative bases, the present method enables us to probe the optimal representation of the reduced dynamics through the joint learning of the non-Markovian features. The constructed features retain a clear physical interpretation and can be loosely viewed as the convolution of the velocity history. This enables us to construct the proper learning formulation such that the reduced dynamics strictly preserves the second fluctuation-dissipation theorem and retains the consistent invariant density distribution. Moreover, the learning process does not require the on-the-fly computation of the time correlations of these features from the time-series samples, 
and automatically ensures numerical stability of the constructed model without empirical treatment. This is particularly well-suited for the construction of reduced dynamics of complex systems, where multi-dimensional resolved variables are often needed to characterize the transition dynamics with non-local cross-correlations among the variables.

Despite the overall success, there are several open questions that are worth further exploration. For instance, the encoders are currently formulated as the standard linear convolution operators. Other nonlinear formulations may further facilitate the retaining of the non-Markovian features. Also, the proper design of the encoder functions that represent state-dependent features may enable us to faithfully construct the reduced dynamics retaining the heterogeneous memory term. We leave these issues for future studies.  


\begin{acknowledgments}
The work is supported by the Extreme Science and Engineering Discovery Environment (XSEDE) Bridges at the Pittsburgh Supercomputing Center through allocation MTH210005, the Strategic Partnership Grant at Michigan State University, and the National Science Foundation under Grant DMS-2110981. 
\end{acknowledgments}


\appendix

\section{Microscale model of the polymer molecule}
\label{app:MD_polymer}
The polymer molecule is modeled as a bead-spring chain consisting of $4$ sub-units. Each sub-unit consists of $4$ atoms. The full potential is given by 
\begin{equation}
    V_{\rm mol}(\mb Q) = \sum_{i\neq j}^{N}V_{\rm p}(Q_{ij})+\sum_{i=1}^{N_b}V_{\rm b}(l_i)+\sum_{i=1}^{N_a}V_{\rm a}(\theta_i)+\sum_{i=1}^{N_d}V_{\rm d}(\phi_i),
\end{equation}
where $V_{\rm p}$,  $V_{\rm b}$, $V_{\rm a}$, and $V_{\rm d}$ represent the pairwise, bond, angle, and dihedral interactions whose detailed forms are specified as below.

The pairwise interaction $V_{\rm p}$ is modeled by the Lennard-Jones potential 
\begin{equation}
V_{\rm p}(Q) = 
\begin{cases} 4\varepsilon \left[ \left(\frac{\sigma}{Q}\right)^{12} - \left(\frac{\sigma}{Q}\right)^6 \right] - 
4\varepsilon \left[ \left(\frac{\sigma}{Q_c}\right)^{12} - \left(\frac{\sigma}{Q_c}\right)^6 \right], \quad Q < Q_c \\
0, \quad Q \ge Q_c
\end{cases}
\end{equation}
where $\varepsilon = 0.005$, $\sigma = 1.8$ and $Q_c = 10.0$.

The bond potential $V_{\rm b}$ is modeled by the finite extensible nonlinear elastic bond (FENE) potential 
\begin{equation}
V_{\rm b}(l) = -\frac{k_s}{2}l^2_{0}\log\left[ 1-\frac{l^2}{l^2_{0}}\right],
\end{equation}
where three different bond types. Within each sub-unit, the atoms $1\mhyphen2$, $3\mhyphen4$ are connected by type-1 bond. The atoms $2\mhyphen3$ are connected by type-2 bond. Finally, the sub-unit groups
are connected by type-3 bond. The detailed parameter set is given by Tab. \ref{tab:bond}. 

\begin{table}[htbp]
\centering
\begin{tabular*}
{0.25\textwidth}{c@{\extracolsep{\fill}} cc}
\hline\hline
Type & $k_s$ & $l_0$   \\
\hline
1    & 0.4  & 1.8    \\
2    & 0.64 & 1.6    \\
3    & 0.32 & 1.8  \\
\hline
\end{tabular*}
\caption{Parameters of the FENE bond interactions.}
\label{tab:bond}
\end{table}

The angle potential $V_{\rm a}$ is modeled by the harmonic angle potential  
\begin{equation}
V_{\rm a}(\theta) = \frac{k_a}{2} \left(\theta-\theta_0\right)^2,
\end{equation}
where two different types. Within each sub-unit group, the bond angles formed by $1\mhyphen2\mhyphen3$
and $2\mhyphen3\mhyphen4$ are imposed by type-1 potential. The bond angles formed by atoms of different sub-unit groups
(e.g., $3\mhyphen4\mhyphen5$, $4\mhyphen5\mhyphen6$) are imposed by type-2 potential.
The detailed parameter set is given by Tab. \ref{tab:angle}. 

\begin{table}[htbp]
\centering
\begin{tabular*}
{0.25\textwidth}{c@{\extracolsep{\fill}} cc}
\hline\hline
Type & $k_a$ & $\theta_0$   \\
\hline
1    & 1.2  & 114.0   \\
2    & 1.5 & 119.7   \\
\hline
\end{tabular*}
\caption{Parameters of the  harmonic angle interaction.}
\label{tab:angle}
\end{table}

The dihedral potential $V_{\rm d}$ is modeled by the multiharmonic dihedral potential
\begin{equation}
V_{\rm d}(\phi) = \sum_{i=1}^6 A_n \cos^{(n-1)}(\phi),
\end{equation}
where two different types. Type-1 dihedral potential is imposed to dihedral 
angles formed by $2\mhyphen3\mhyphen4\mhyphen5$, $4\mhyphen5\mhyphen6\mhyphen7$, $\cdots$.
Type-2 dihedral potential is imposed to dihedral 
angles formed by $3\mhyphen4\mhyphen5\mhyphen6$, $7\mhyphen8\mhyphen9\mhyphen10$, $\cdots$. The detailed parameter set is given by Tab. \ref{tab:dihedral}. 

\begin{table}[htbp]
\centering
\begin{tabular*}
{0.7\textwidth}{c@{\extracolsep{\fill}} cccccc}
\hline\hline
Type & $A_1$ & $A_2$ & $A_3$ & $A_4$  & $A_5$ & $A_6$ \\
\hline
1    & 0.0673 &1.8479 &0.0079  &-2.2410  &-0.0058   &0.0051 \\
2    & 0.1602  &-3.9993   &0.2483   &6.2837   &0.0165  &-0.0146  \\
\hline
\end{tabular*}
\caption{Parameters of the multiharmonic dihedral interaction.}
\label{tab:dihedral}
\end{table}

\section{Construction of the four-dimensional free energy function}
\label{app:free_energy}
Accurate construction of the multi-dimensional free energy is a well-known non-trivial problem. To construct the free energy function $U(\mb q)$ for the four-dimensional resolved variables $\mb q$ defined by \eqref{eq:resolved_variables}, we conduct the constraint molecular dynamics simulation to sample the average force. Specifically, for each target configuration $\mb q^{\ast}$, we impose a biased quadratic potential $U_{\rm bias} (\mb q, \mb q^{\ast})$ by
\begin{equation}
U_{\rm bias} (\mb q, \mb q^{\ast}) = \frac{1}{2} \sum_{i=1}^4 k_i\left(q_i - q_i^{\ast} \right)^2, 
\end{equation}
where $k_1, \cdots, k_4$ represents the magnitude of the bias potential. We choose the values such that the fluctuations are about $5\%$ of target values. For the polymer molecule considered in the present study, the effective constraint force applied to the full atom $\left\{\mb Q_j\right\}_{j=1}^N$ is given by
\begin{equation}
\mb F_{\rm bias} (\mb q, \mb q^{\ast}) = - \sum_{i=1}^4 k_i\left(q_i - q_i^{\ast} \right) \nabla_{\mb Q_j} q_i,   
\end{equation}
where the gradient terms are given by
\begin{equation}
\begin{split}
\nabla_{\mb Q_j} q_1 &= \frac{\mb Q_1 - \mb Q_N}{q_1} \delta_{j,1} +   \frac{\mb Q_N - \mb Q_1}{q_1}  \delta_{j,N}, \\
\nabla_{\mb Q_j} q_2 &= \frac{2\left(\mb Q_j - \mb Q_c\right)}{N q_2}, \\
\nabla_{\mb Q_j} q_3 &= \frac{\mb Q_1 - \mb Q_{\lfloor{\frac{N}{2}}\rfloor}}{q_3} \delta_{j,1} +   \frac{\mb Q_{\lfloor{\frac{N}{2}}\rfloor} - \mb Q_1}{q_3}  \delta_{j,\lfloor{\frac{N}{2}}\rfloor}, \\
\nabla_{\mb Q_j} q_4 &= \frac{\mb Q_N - \mb Q_{\lceil{\frac{N}{2}}\rceil}}{q_4} \delta_{j,N} +   \frac{\mb Q_{\lceil{\frac{N}{2}}\rceil} - \mb Q_N}{q_4} \delta_{j,\lceil{\frac{N}{2}}\rceil},
\end{split}    
\end{equation}
where $\delta_{i,j}$ represents the Kronecker delta function. 

The free energy $U(\mb q)$ is approximated by a $4\mhyphen$layer fully connected neural network $\tilde{U}(\mb q)$. Each hidden layer has $160$ neurons; hyperbolic tangent function is used as the activation function. $\tilde{U}(\mb q)$ is trained by minimizing the empirical loss 
\begin{equation}
L = \sum_{k=1}^{N_s} \left\Vert -\nabla_{\mb q^{(k)}} \tilde{U}(\mb q) -  \mb F_{\rm bias} (\mb q, \mb q^{(k)})\right\Vert^2, 
\end{equation}
where $\mb q^{(k)}$ represents a sampled configuration. In this work, we construct $\tilde{U}(\mb q)$ using $N_s = 400000$ sample points. 

To verify the accuracy of $\tilde{U}(\mb q)$, we numerically evaluate the integration
\begin{equation}
k_BT \mb I \equiv \int \mb q \otimes \nabla U(\mb q) {\rm e}^{-U(\mb q)/k_BT} \diff \mb q / \int {\rm e}^{-U(\mb q)/k_BT} \diff \mb q  \approx \frac{1}{N_s}\sum_{k=1}^{N_s} \mb q^{(k)} \otimes \nabla \tilde{U}(\mb q^{(k)}).     
\end{equation}
Therefore, the difference between the numerical summation and $k_BT \mb I$ provide a metric. For this case, $k_BT = 1$. The average term yields
\begin{equation}
\frac{1}{N_s}\sum_{k=1}^{N_s} \mb q^{(k)} \otimes \nabla \tilde{U}(\mb q^{(k)}) = 
\begin{bmatrix}
1.0362  & -0.0011 &0.0087 &0.0062\\
0.0094 & 0.9814 &0.0021 &0.0018\\
0.0096 & 0.0068 &0.9913 & -0.0020\\
0.0076&0.0098 &0.0008 & 0.9913\\
\end{bmatrix},    
\end{equation}
which verifies that the constructed $\tilde{U}(\mb q)$ is an accurate approximation of $U(\mb q)$.

\section{Fluctuation-dissipation theorem of the extended dynamics}
\label{app:fdt}
For the extended dynamics in form of Eqs. \eqref{eq:SDE_simple}\eqref{eq:G_matrix}, we can show that the embedded memory kernel $\tilde{\bm \theta}(t)$ and fluctuation term $\Tilde{\mb{\mathcal{R}}}(t)$ satisfy the second-fluctuation dissipation theorem. Without loss of generality, we set the covariance of the non-Markovian features to be $k_BT \mb I$ following the learning method presented in Sec. \ref{sec:learn_reduced_dynamics}, i.e., $\bm \Lambda = \mb I$, $\tilde{\mb J} = \mb J$. 
\begin{proposition} 
The embedded memory kernel of the extended dynamics \eqref{eq:SDE_simple}\eqref{eq:G_matrix} takes the form $\displaystyle \tilde{\bm \theta}(t) = -\left(\mb J_{11} \delta(t) + \mb J_{12} {\rm e}^{\mb J_{22} t} \mb J_{21}\right)$. Furthermore, by choosing the initial condition of $\bm \zeta$ and the white noise term $\bm\xi (t) = \bm \Sigma \dot{\mb W}_t$ satisfying
\begin{equation}
\begin{split}
\left\langle \bm \zeta(0) \bm\zeta(0)^T \right\rangle &= \beta^{-1} \mb I \\
\left\langle \bm \xi (t) \bm\xi (s)^T \right\rangle &= -\beta^{-1} (\mb J  +  \mb J^T) \delta(t-s),
\end{split}
\label{eq:sde_noise_condition}
\end{equation}
the embedded kernel $\tilde{\bm\theta}(t)$ and $\mb{\mathcal{R}}(t)$ satisfies the second fluctuation-dissipation theorem, i.e., 
\begin{equation}
\left\langle \Tilde{\mb{\mathcal{R}}}(t) \Tilde{\mb{\mathcal{R}}}(t')^T \right\rangle = -\beta^{-1} \left(\tilde{\mb J}_{12} {\rm e}^{\tilde{\mb J}_{22}(t-t')} \tilde{\mb J}_{21} + (\tilde{\mb J}_{11} + \tilde{\mb J}_{11}^T) \delta(t - t')\right).
\end{equation}
\end{proposition}

\begin{proof}
With $\bm \Lambda = \mb I$ and $\tilde{\mb J} = \mb J$, we can take the integration of $\bm\zeta(t)$ in Eq. \eqref{eq:SDE_simple}, yielding
\begin{equation}
\bm\zeta(t) = \int_0^t {\rm e}^{\mb J_{22}(t-s)} \mb J_{21} \mb v(s) \diff s + 
\int_0^t {\rm e}^{\mb J_{22}(t-s)}  \bm \xi_2(s) \diff s
+ {\rm e}^{\mb J_{22}t} \bm\zeta(0).
\end{equation}
Plugging $\bm\zeta(t)$ into the dynamic equation of $\mb v$ gives
\begin{equation}
\begin{split}
\mb M\dot{\mb v} = &-\nabla U(\mb q) + \mb J_{11}  \mb v+ \int_0^t \mb J_{12} {\rm e}^{\mb J_{22}(t-s)} \mb J_{21} \mb v(s) \diff t  \\
&+\underbrace{\bm \xi_1(t)}_{{\Tilde{\mathcal{R}}}_1(t)} 
+\underbrace{\int_0^t \mb J_{12} {\rm e}^{\mb J_{22}(t-s)}  \bm \xi_2(s) \diff s}_{{\Tilde{\mathcal{R}}}_2(t)}
+ \underbrace{\mb J_{12}  {\rm e}^{\mb J_{22}t} \bm\zeta(0)}_{{\Tilde{\mathcal{R}}}_3(t)}.  
\end{split}
\label{eq:extended_int}
\end{equation}
We check the covariance matrices of the noise terms, i.e., 
\begin{equation}
\begin{split}
\left\langle {\Tilde{\mathcal{R}}}_1(t) {\Tilde{\mathcal{R}}}_1(t')^T\right\rangle &= -\beta^{-1}(\mb J_{11} + \mb J_{11}^T) \delta(t - t'), \\
\left\langle {\Tilde{\mathcal{R}}}_2(t) {\Tilde{\mathcal{R}}}_2(t')^T\right\rangle &= 
\int_0^t \int_0^{t'} \mb J_{12} {\rm e}^{\mb J_{22}(t-s)}  \left\langle \bm \xi_2(s) \bm \xi_2(s')^T\right\rangle  {\rm e}^{\mb J_{22}^T(t'-s')} \mb J_{12}^T  \diff s \diff s' \\
&=  -\beta^{-1}\int_0^t\int_0^{t'} \mb J_{12} {\rm e}^{\mb J_{22}(t-s)}  (\mb J_{22} + \mb J_{22}^T)\delta(s-s')   \mb J_{21}^T {\rm e}^{\mb J_{22}^T(t'-s')} \mb J_{12}^T \diff s  \diff s' \\
&= -\beta^{-1}\int_0^{t'} \mb J_{12} {\rm e}^{\mb J_{22}(t-s')}  (\mb J_{22} + \mb J_{22}^T)  \mb {\rm e}^{\mb J_{22}^T(t'-s')} \mb J_{12}^T  \diff s', \\
&= -\beta^{-1} \mb J_{12} {\rm e}^{\mb J_{22}t + \mb J_{22}^T t'} \mb J_{12}^T + \beta^{-1} \mb J_{12} {\rm e}^{\mb J_{22} (t-t')}  \mb J_{12}^T, \quad \forall t' \le t \\
\left\langle {\Tilde{\mathcal{R}}}_3(t) {\Tilde{\mathcal{R}}}_3(t')^T\right\rangle &= \mb J_{12} {\rm e}^{\mb J_{22}t}
\left\langle \bm \zeta(0) \bm\zeta(0)^T \right\rangle {\rm e}^{\mb J_{22}^T t'}\mb J_{12}^T \\
&= \beta^{-1} \mb J_{12} {\rm e}^{\mb J_{22}t} {\rm e}^{\mb J_{22}^T t'}\mb J_{12}^T. 
\end{split}
\label{eq:cov_noise_1}
\end{equation}
Moreover, for $t > t'$, all the cross terms vanish except $\left\langle {\Tilde{\mathcal{R}}}_2(t) {\Tilde{\mathcal{R}}}_1(t')^T\right\rangle$, i.e.,
\begin{equation}
\begin{split}
\left\langle {\Tilde{\mathcal{R}}}_2(t) {\Tilde{\mathcal{R}}}_1(t')^T\right\rangle &= \int_0^t \mb J_{12} {\rm e}^{\mb J_{22}(t-s)}  \left\langle \bm \xi_2(s) \bm \xi_1(t') \right\rangle \diff s  \\
&= -\beta^{-1} \int_0^t \mb J_{12} {\rm e}^{\mb J_{22}(t-s)} (\mb J_{21} + \mb J_{12}^T) \delta(t' -s )\diff s \\
&= -\beta^{-1} \mb J_{12} {\rm e}^{\mb J_{22}(t-t')} (\mb J_{21} + \mb J_{12}^T).
\end{split}
\label{eq:cov_noise_2}
\end{equation}
Combining Eq. \eqref{eq:cov_noise_1} and Eq. \eqref{eq:cov_noise_2}, we have
\begin{equation}
\begin{split}
\left\langle {\Tilde{\mathcal{R}}}(t) {\Tilde{\mathcal{R}}}(t')^T\right\rangle &= \beta^{-1} \mb J_{12} {\rm e}^{\mb J_{22} (t-t')}  \mb J_{12}^T -\beta^{-1} \mb J_{12} {\rm e}^{\mb J_{22}(t-t')} (\mb J_{21} + \mb J_{12}^T) \\
&\quad - \beta^{-1}(\mb J_{11} + \mb J_{11}^T) \delta(t - t') \\
&= -\beta^{-1} \left(\mb J_{12} {\rm e}^{\mb J_{22}(t-t')} \mb J_{21} + (\mb J_{11} + \mb J_{11}^T) \delta(t - t')\right).   
\end{split}
\label{eq:cov_noise_3}
\end{equation}
\end{proof}
As a special case, by imposing the constraint specified by Eq. \eqref{eq:J_SPD_2} such that $\mb J_{11} + \mb J_{11}^T = 0$ and $\mb J_{12} = -\mb J_{21}^T$, the memory kernel $\tilde{\bm\theta}(t)$ recovers $-\mb J_{12} {\rm e}^{\mb J_{22}t} \mb J_{12}^T$ without the Markovian part, and the second fluctuation-dissipation theorem recovers the standard form, i.e., 
\begin{equation}
\left\langle \Tilde{\mb{\mathcal{R}}}(t) \Tilde{\mb{\mathcal{R}}}(0)^T \right\rangle = \beta^{-1} \tilde{\bm\theta}(t).    
\end{equation}

\section{Invariant probability density function}
\label{app:invariant_density}
\begin{proposition} 
By choosing the white noise following Eq. \eqref{eq:sde_noise_condition}, the reduced model \eqref{eq:SDE_simple}\eqref{eq:G_matrix} retains the invariant density function
\begin{equation}
\rho_{\rm eq}(\mb q, \mb p, \bm\zeta) =  \exp\left[-\beta W(\mb q, \mb p, \bm \zeta)\right]  /
\int \exp\left[-\beta W(\mb q, \mb p, \bm \zeta)\right] \diff \mb q \diff \mb p \diff \bm\zeta.
\end{equation}
\end{proposition}

\begin{proof}
By Eq. \eqref{eq:sde_noise_condition}, the covariance of the white noise of the full extended system is given by $\mb G + \mb G^T = {\rm diag}(0, \bm\Sigma \bm\Sigma^T)$. Accordingly, the Fokker-Plank equation follows
\begin{equation}
\frac{\partial \rho(\mb z, t)}{\partial t} = \nabla \cdot \left(-\mb G \nabla W(\mb z) \rho(\mb z, t) - \frac{1}{2} \beta^{-1}(\mb G+ \mb G^T) \nabla \rho(\mb z, t) \right), 
\end{equation}
where $\rho(\mb z, t)$ represents the probability density function of the extended variables $\mb z = [\mb q; \mb p; \bm \zeta]$. 
For $\rho_{\rm eq}(\mb q, \mb p, \bm\zeta) \propto  \exp\left[-\beta W(\mb q, \mb p, \bm \zeta)\right]$, the RHS follows
\begin{equation}
\begin{split}
\nabla \cdot \left(\beta^{-1} \mb G \nabla \rho_{\rm eq}(\mb z, t)   - \frac{1}{2} \beta^{-1}(\mb G+ \mb G^T) \nabla \rho_{\rm eq}(\mb z, t) \right) 
&=  \beta^{-1} \nabla \cdot  \left({\mb G}^A   \nabla \rho_{\rm eq}(\mb z, t)\right) \\
&\equiv 0,
\end{split}
\end{equation}
where the last identity holds because $\mb G^A$ is anti-symmetric. 
\end{proof}


%

\end{document}